\documentclass[10pt,conference,letterpaper]{IEEEtran}
\usepackage{color}
\usepackage{booktabs} 

\usepackage{subfig}

\usepackage{todonotes}

\usepackage{listings}
\usepackage[linesnumbered,ruled]{algorithm2e}
\usepackage{multirow}

\SetAlFnt{\small}
\SetAlCapFnt{\small}
\SetAlCapNameFnt{\small}
\SetAlCapHSkip{0pt}
\IncMargin{-\parindent}
\SetKwInOut{Input}{input}
\SetKwInOut{Output}{output}
\SetKwInOut{SideEffect}{side effect}
\SetKwFor{For}{for}{do}{}

\usepackage{enumitem}
\setlist[itemize]{leftmargin=*}

\usepackage{amsmath}
\usepackage{mathtools}
\DeclarePairedDelimiter{\ceil}{\lceil}{\rceil}

\DeclareMathOperator*{\argmin}{arg\,min}

\usepackage{amsthm}
\newtheorem{lemma}{Lemma}
\newtheorem{theorem}{Theorem}

\usepackage{comment}
\usepackage{url}

\usepackage{lmodern}
\begin{document}
	\title{VEBO: A Vertex- and Edge-Balanced Ordering Heuristic to Load Balance Parallel Graph Processing}
	\author{\IEEEauthorblockN{Jiawen Sun}
		\IEEEauthorblockA{Queen's University Belfast\\
			Email: jsun03@qub.ac.uk}
		\and
		\IEEEauthorblockN{Hans Vandierendonck}
		\IEEEauthorblockA{Queen's University Belfast\\
			Email: h.vandierendonck@qub.ac.uk}
		\and
		\IEEEauthorblockN{Dimitrios S.\ Nikolopoulos}
		\IEEEauthorblockA{Queen's University Belfast\\
			Email: d.nikolopoulos@qub.ac.uk}}
	\maketitle
	\begin{abstract}
%
%
Graph partitioning drives graph processing
in distributed, disk-based and NUMA-aware systems.
A commonly used partitioning goal is to balance the number
of edges per partition in conjunction with minimizing
the edge or vertex cut.
While this type of partitioning is computationally expensive,
we observe that such topology-driven partitioning
nonetheless results in computational load imbalance.

We propose
Vertex- and Edge-Balanced Ordering (VEBO):
balance the number of edges and the number of unique
destinations of those edges.
VEBO optimally balances edges and vertices for graphs with a power-law
degree distribution.
Experimental evaluation on three
shared-memory graph processing systems (Ligra, Polymer and GraphGrind)
shows that VEBO achieves excellent load balance and improves
performance
by 1.09x over Ligra,
1.41x over Polymer and 1.65x over GraphGrind,
compared to their respective partitioning algorithms,
averaged across 8 algorithms and
7 graphs.


	\end{abstract}
	
	%
	%

	

	\section{Introduction}
	
	Graph partitioning is used extensively to orchestrate parallel execution
	of graph processing
	in distributed systems~\cite{gonzalez2012powergraph},
	disk-based processing~\cite{kyrola2012graphchi,roy2013x} 
	and NUMA-aware shared memory systems~\cite{zhang2015numa,sun:17:ics}.
	In order to maximize processing speed,
	each partition should take the same amount of processing time.
	Moreover, the partitions should be largely independent to
	minimize the volume of data communication.
	It has been demonstrated that partitioning the edge set is more effective
	than partitioning the vertex set~\cite{gonzalez2012powergraph},
	leading to the commonly used heuristic to balance edge counts
	and minimize vertex replication~\cite{gonzalez2012powergraph,bourse:14:bep,li:17:spac,kyrola2012graphchi}.
	These constraints are typically mutually incompatible for
	scale-free graphs, resulting in a compromise between edge balance
	and vertex replication~\cite{bourse:14:bep}.
	
	%
	We have observed, however, that edge balance
	does not uniquely determine execution time.
	For instance, for an approximating PageRankDelta~\cite{Ligra} computation,
	during a first phase of the algorithm,
	about half of low-degree vertices converge before
	any high-degree vertex converges.
	A partition that consists of mostly high-degree vertices will thus take
	longer to process than a partition with only low-degree vertices, resulting
	in load imbalance.
	Note that it is likely to encounter partitions with mostly low-degree
	vertices
	in graphs with a power-law degree distribution as these
	graphs have many more low-degree vertices than high-degree vertices.
	
	The key contribution of this paper is to identify that
	the time for processing a graph partition depends on both
	the number of edges and the number of unique destinations in that partition.
	This presents a new heuristic to partition graphs through
	joint destination vertex- and edge-balanced partitioning,
	which we call VEBO. 
	
	
	A key motivation for considering joint vertex and edge balancing
	is provided by the classification of Sun~\emph{et al},
	who observed a distinction between \emph{edge-oriented} algorithms
	and \emph{vertex-oriented} algorithms~\cite{sun:17:ics}.
	Edge-oriented algorithms, like PageRank, perform an amount of computation
	proportional to the number of edges.
	In contrast,
	vertex-oriented algorithms, like Breadth-First Search, perform an
	amount of computation proportional to the number of vertices.
	These properties strongly affect partitioning:
	a vertex-balanced partition can result in
	almost a 40\% speedup compared to edge-balanced partitioning
	for vertex-oriented algorithms~\cite{sun:17:ics}.
	While Sun~\emph{et al} selected the partitioning heuristic depending on the
	algorithm, VEBO seamlessly resolves
	this important distinction between algorithm types.
	
	
	A second important contribution of this work is to identify a need
	to adapt vertex order to the characteristics of the graph processing
	system. Each system has unique design choices, which determine
	its key performance bottlenecks.
	For instance, Ligra~\cite{Ligra} uses dynamic scheduling to manage
	parallelism, but does itself not improve memory locality.
	In contrast, Polymer~\cite{zhang2015numa} and GraphGrind~\cite{sun:17:ics}
	perform NUMA-aware data layout and
	use static scheduling in order to bind computation to the appropriate
	NUMA domains. Static scheduling makes parallel loops sensitive to
	load balance as the execution time of the loop is determined
	by the last-completing thread.
	It may thus be expected that vertex ordering
	serves different purposes: for Ligra, memory locality should be
	improved, while for Polymer and GraphGrind, load balance is more important.
	
	
	
	We propose a graph partitioning algorithm that calculates an optimally
	load-balanced
	partition for power-law graphs with time complexity $\mathcal{O}(n\log P)$,
	where $n$ is the number of vertices in the graph
	and $P$ is the number of partitions.
	Extensive experimental evaluation using three shared memory graph processing
	systems demonstrates a near-perfect computational load balance across a variety
	of graph data sets and algorithms.
	Contrary to heuristics such as
	Reverse Cuthill-McKee (RCM)~\cite{george:94:sparse} and Gorder~\cite{wei:16:gorder}
	that aim to optimize memory locality,
	we obtain a consistent performance improvement when processing seven
	scale-free graphs.
	
	In summary, this paper makes the following contributions:
	\begin{enumerate}
		\item Demonstrating the need to balance the number of unique destinations along with the number of edges in order to achieve computational load balance
		\item A simple vertex reordering algorithm that optimally balances both edges and unique destinations using time proportional to $\mathcal{O}(n\log P)$
		\item Addressing vertex-oriented and edge-oriented algorithms
		using a single graph partitioning heuristic
		\item Extensive experimental evaluation
		using three shared memory graph processing systems (Ligra, Polymer
		and GraphGrind) and a comparison
		to edge balancing
	\end{enumerate}
	
	The remainder of this paper is organized as follows:
	Section~\ref{sec:motiv} motivates the load balancing heuristic.
	Section~\ref{sec:reorder} presents the VEBO algorithm and proves its optimality.
	Section~\ref{sec:evalm} presents our experimental evaluation methodology.
	The experimental evaluation of VEBO is presented in Section~\ref{sec:eval}.
	Related work is discussed in Section~\ref{sec:rela}.
	\begin{algorithm}[!t]
		\Input{Graph $G=(V,E)$; number of partitions $P$}
		\Output{Graph partitions $G_i=(V,E[i])$ for $i=0,\ldots,P-1$}
		
		Let avg = $|E|/P$\tcp*{target edges per partition}
		Let E[P] = \{ 0 \}\tcp*{array of length $P$, initialized to 0}
		Let i = 0\;
		\For{$t\leftarrow 0$ \KwTo $|V|-1$}{
			Let $v = v_{t}$\tcp*{the $t$-th vertex}
			\uIf{$|E[i]|$ $\ge$ avg \textbf{and} i $< P-1$}{
				++i\tcp*{E[i] has exceeded target edges}
			}
			E[i] = E[i] $\ \cup$ in-edges($v$)\tcp*{i is home partition of v}
		}
		\caption{Locality-preserving edge-balanced partitioning of the destination vertices}
		\label{algo:part}
	\end{algorithm}
	
	\begin{figure}[t]
		\centering	
		\subfloat[Twitter]{
			\includegraphics[width=.48\columnwidth,clip]{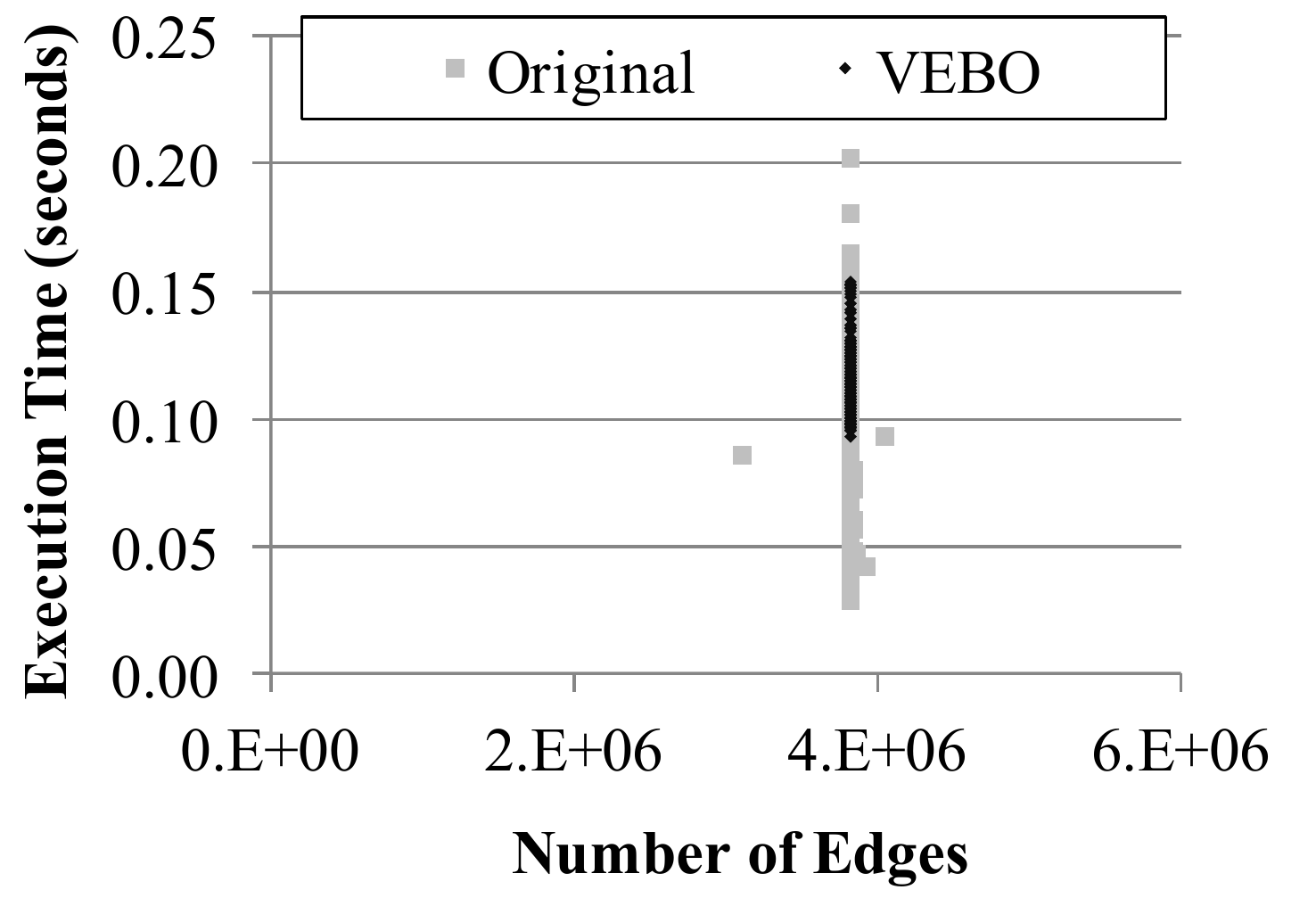}
			\label{fig:twitter-edge}
		}
		\subfloat[Friendster]{
			\includegraphics[width=.49\columnwidth,clip]{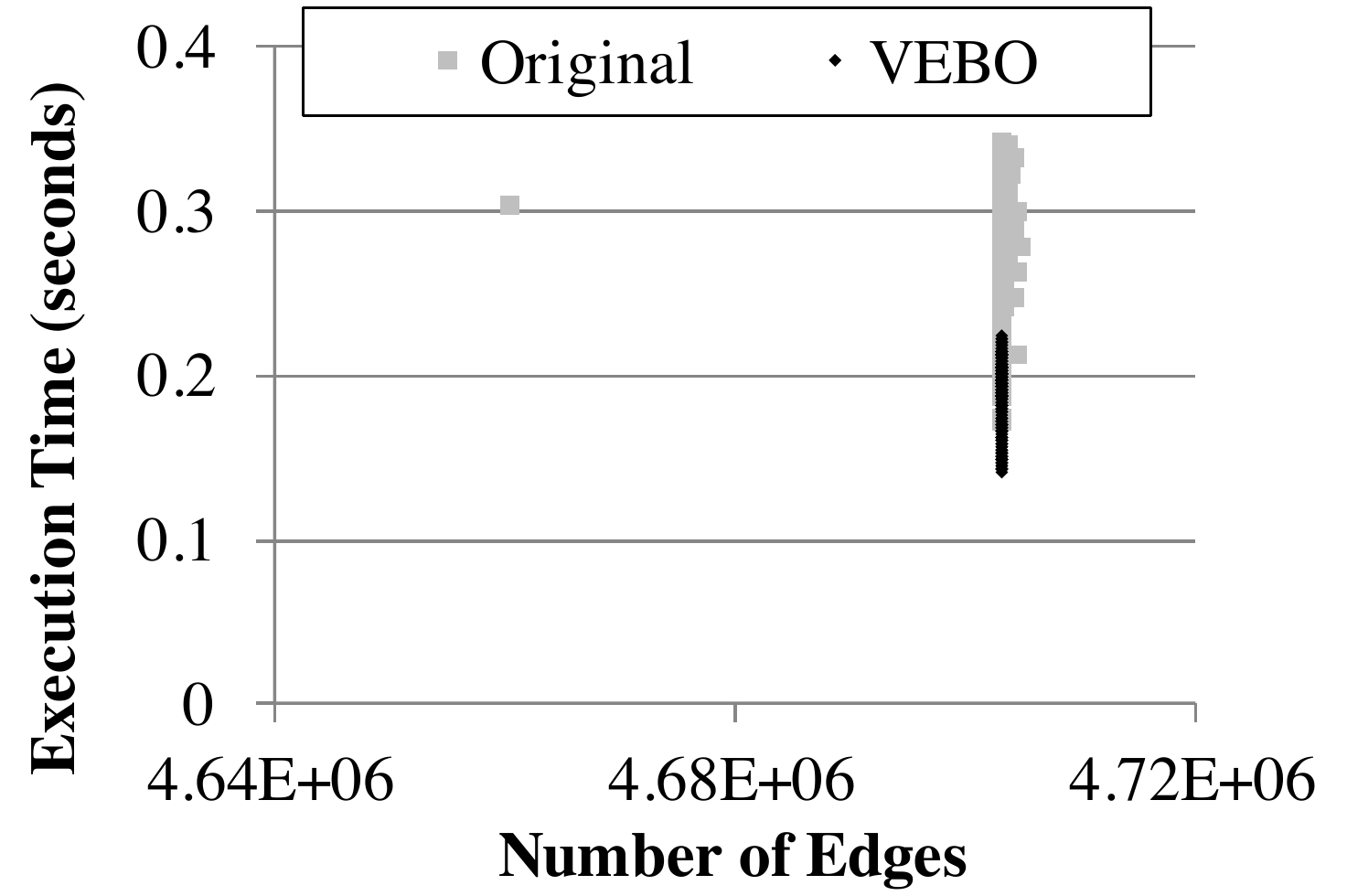}
			\label{fig:fd-edge}
		}
		
		\subfloat[Twitter]{
			\includegraphics[width=.49\columnwidth,clip]{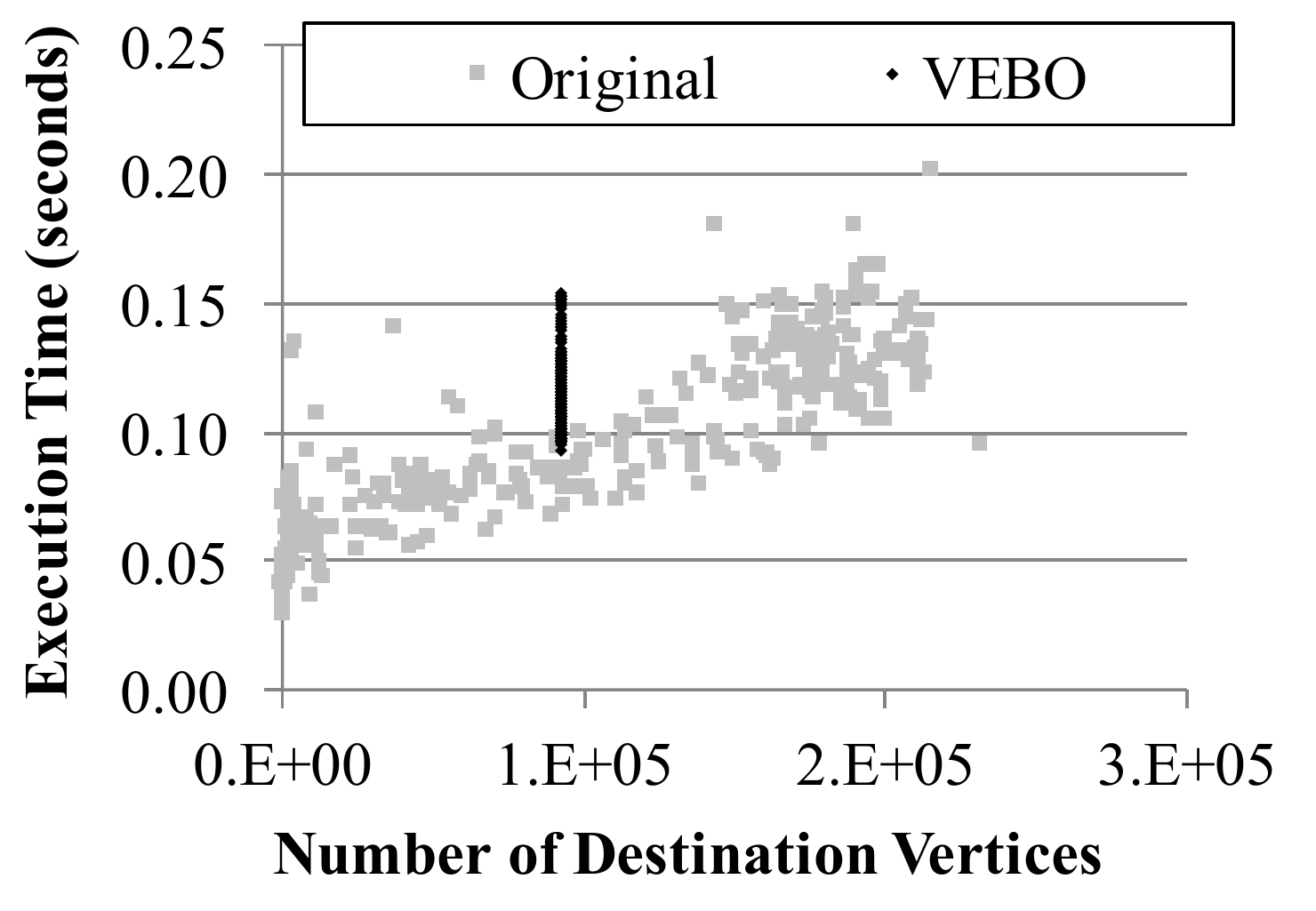}
			\label{fig:twitter-dest}
		}
		\subfloat[Friendster]{
			\includegraphics[width=.49\columnwidth,clip]{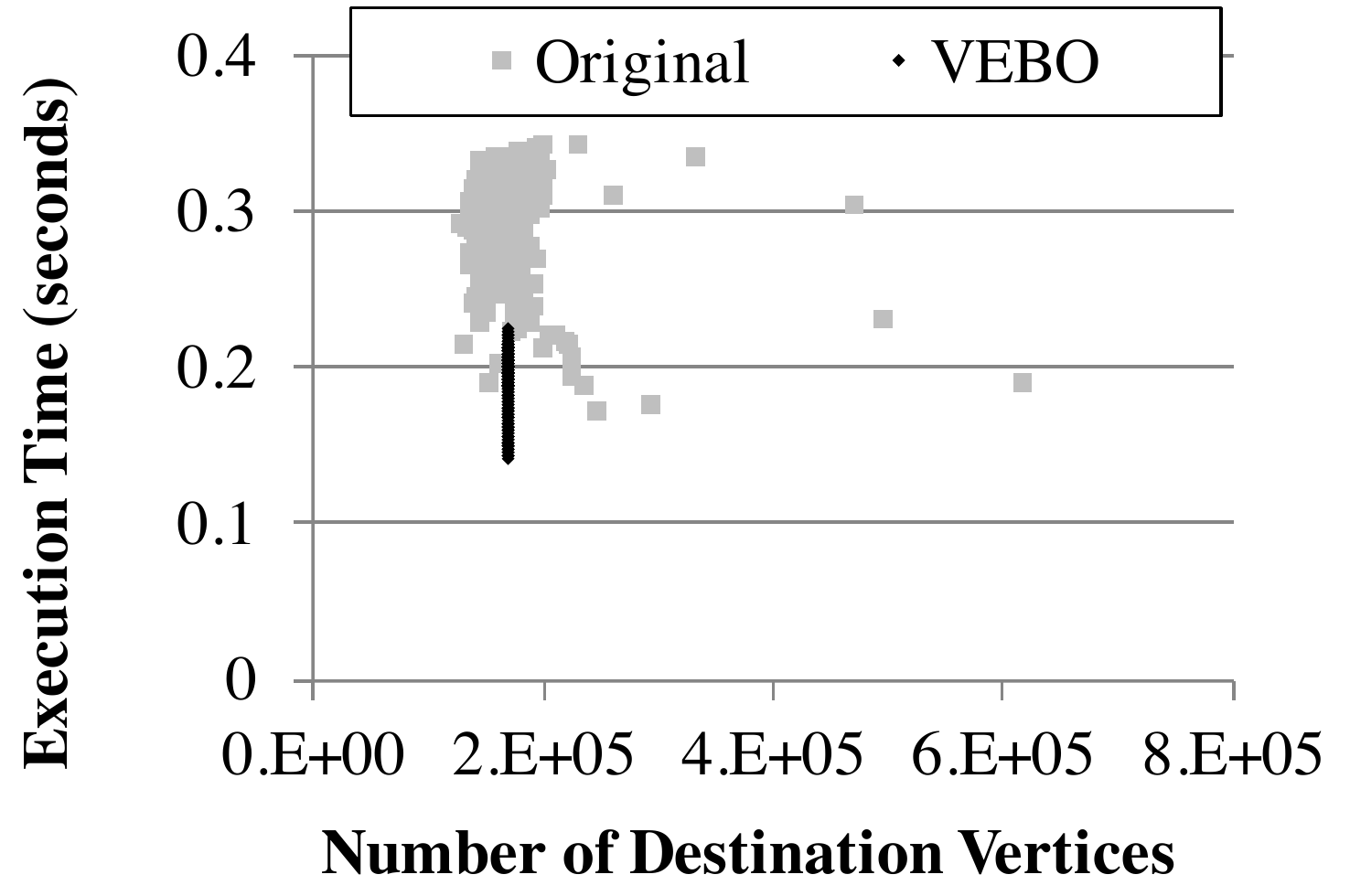}
			\label{fig:tfd-dest}
		}
		
		\subfloat[Twitter]{
			\includegraphics[width=.49\columnwidth,clip]{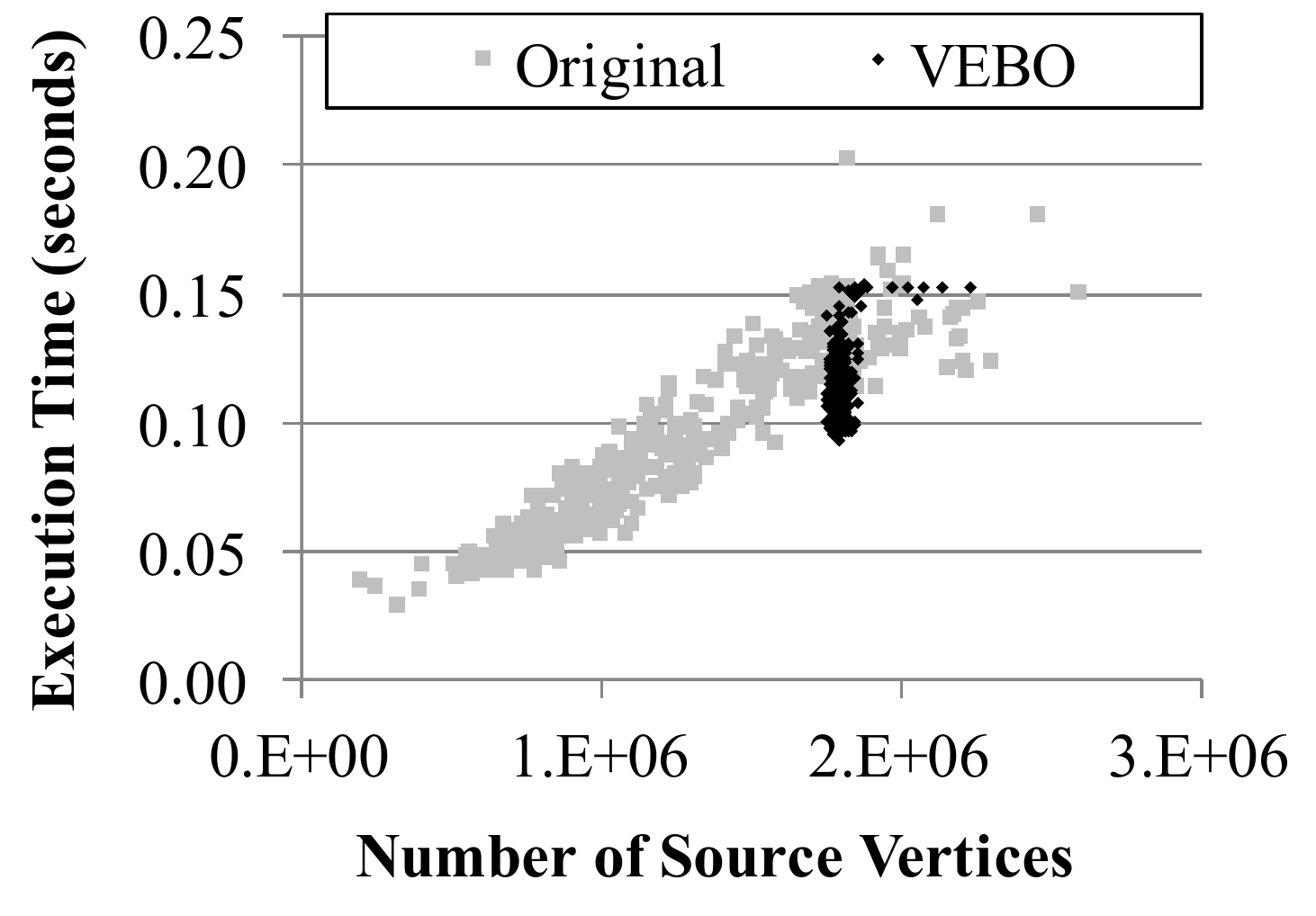}
			\label{fig:ftwitter-src}
		}
		\subfloat[Friendster]{
			\includegraphics[width=.49\columnwidth,clip]{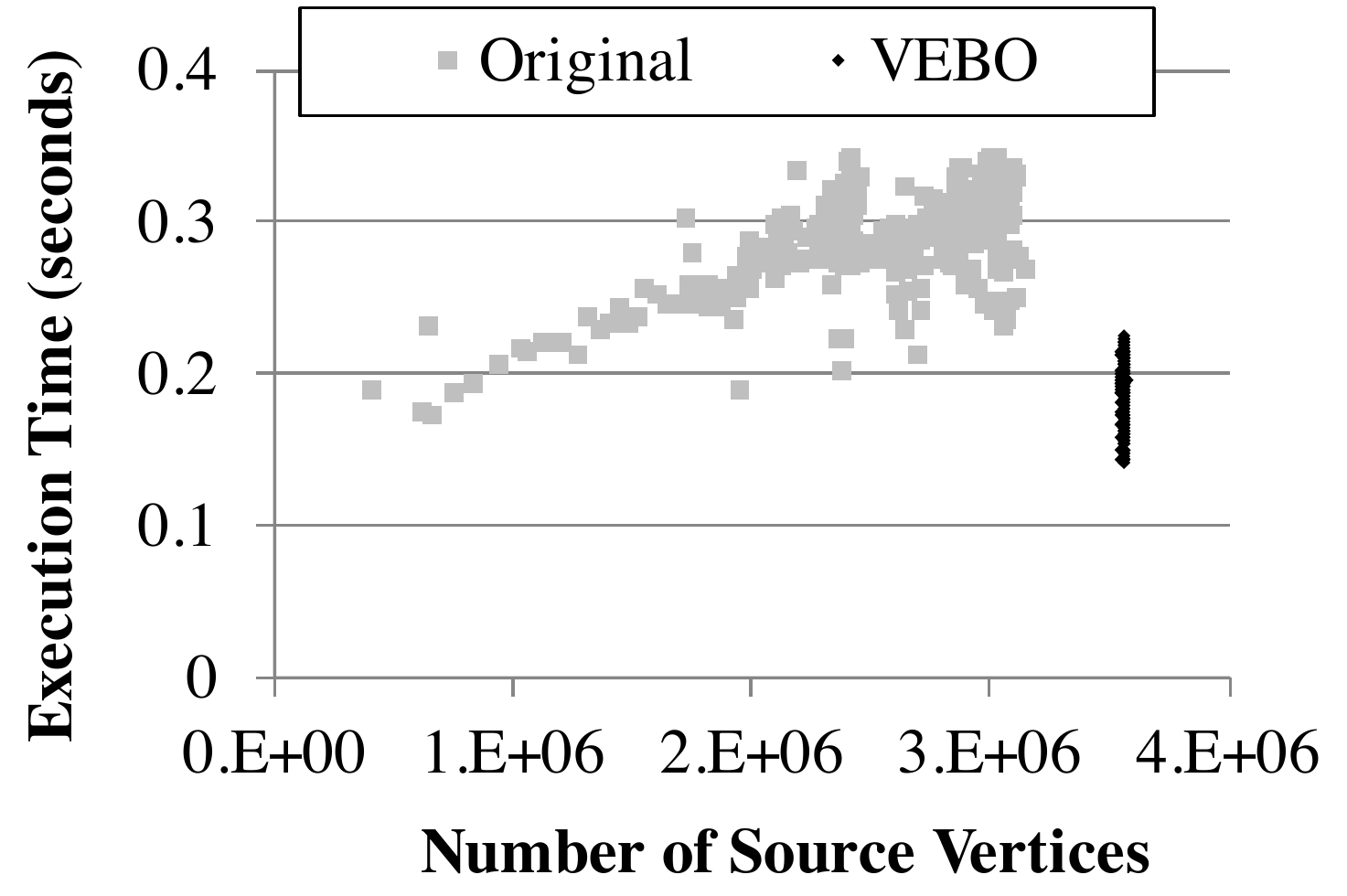}
			\label{fig:fd-src}
		}
		\caption{Processing time of a partition as a function of the number of edges, destinations, and source vertices in the partition. Each data point corresponds to one of 384 partitions.  Average time of Twitter in original is 0.119s and in VEBO is 0.111s. Average time of Friendster in original is 0.289s and in VEBO is 0.234s. }
		\vspace*{-5mm}
		\label{fig:motiv}
	\end{figure}	
	\section{Motivation}
	\label{sec:motiv}
	The edge-balance heuristic is commonly used
	to balance the computation.
	Algorithm~\ref{algo:part} shows a simple, locality-preserving
	edge balancing algorithm for graph partitioning.
	We call it \emph{partitioning by destination} as edges are assigned to the
	partition that holds their destination vertex.
	The algorithm is locality-preserving in the sense that
	each partition consists of a chunk of consecutively numbered vertices.
	This algorithm is used in disk-based~\cite{kyrola2012graphchi}
	and NUMA-aware graph processing~\cite{zhang2015numa,sun:17:ics}.
	
	Figure~\ref{fig:motiv} shows the processing time for each of 384 partitions
	when executing one iteration of the PageRank algorithm.
	The graph is represented using the coordinate format (COO)
	and edges are sorted in the access order of a Hilbert space filling curve
	in order to improve memory locality~\cite{murray:13:naiad,sun:17:icpp}.
	Each partition is processed sequentially by one thread.
	We study the Twitter and Friendster graphs.
	Details on the graphs and experimental setup are
	provided in Section~\ref{sec:evalm}.
	
	The top two plots in Figure~\ref{fig:motiv} show that Algorithm~\ref{algo:part}
	achieves good edge balance. There is some variation on the number
	of edges in each partition, which results from high-degree vertices
	that appear at the boundary between two partitions.
	Placing a high-degree vertex in
	a first partition will overload it, placing it in the next
	partition will leave the first underloaded.
	While partitions are edge-balanced well, the execution time per partition
	varies over a factor of 6.9x and 2x for the
	Twitter and Friendster graphs, respectively.
	The VEBO heuristic, which we will present in the next Section,
	reduces to variation to 1.6x (Twitter) and 1.4x (Friendster).
	
	
	The plots moreover show that the processing time of a partition
	is correlated to the number 
	of destination vertices (middle row), and of source vertices (bottom).
	Partitions with few destination vertices
	(and thus holding vertices with a high in-degree) are processed faster
	than a partition holding many low-degree vertices.
	
	
	While both the number of unique source and destination vertices
	in a partition affects processing time, we choose to balance the
	number of destinations. Balancing both source and destination numbers
	would be as computationally complex as minimizing edge cut and it
	is our express goal to minimize the time taken for partitioning.
	A number of graph processing systems
	partition the destination set
	to create parallelism and
	avoid data races~\cite{kyrola2012graphchi,zhang2015numa,sun:17:ics}.
	Partitioning by destination ensures data race freedom as graph algorithms
	typically update values associated to the destinations of edges.
	Others partition the source vertices~\cite{roy2013x}. 
	For these systems, the analogous balancing criteria would focus
	on the source vertices.
	
	\section{The VEBO Algorithm}
	\label{sec:reorder}
	\subsection{Problem Statement}
	Assume a graph $G=(V,E)$ with power-law in-degree distribution.
	Let $N$ be one more than the highest in-degree in the graph.
	Let $n$ be the number of vertices and let $m$ be the number of
	vertices with non-zero in-degree.
	
	We model the in-degree distribution using a Zipf distribution
	where 
	$s\ge 0$
	is the real-valued exponent governing the skewedness
	of the degree distribution,\footnote{The exponent $s$ is related to the
		exponent $\alpha$ in the power-law degree distribution $p_k=\beta k^{-\alpha}$
		by $\alpha=1+1/s$.}
	$N$ is the number of ranks
	and $p_k$, $k=1,\ldots,N$ is the probability that a vertex has degree $k-1$:
	\begin{equation}
	p_k = \frac{k^{-s}}{H_{N,s}}
	\end{equation}
	\begin{figure}[t]
		\centerline{\includegraphics[width=.95\columnwidth]{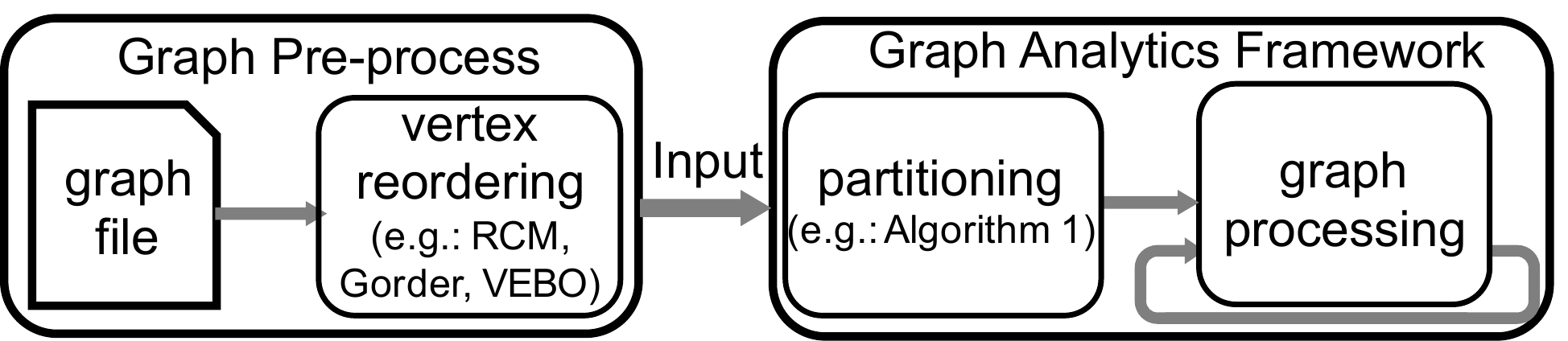}}
		\caption{A graph processing pipeline with vertex reordering 
			and graph partitioning}
		\label{fig:steps}
	\end{figure}
	where $H_{N,s}=\sum_{i=1}^N i^{-s}$ is a Generalized Harmonic Number.
	As such, the most frequent in-degree in the graph is zero,
	and least frequent in-degree is $N-1$.
	We make no assumptions about the out-degree distribution.
	
	VEBO partitions the vertex set in $P$ parts
	such that $V=\cup_{i=0}^{P-1} V_i$ and $V_i\cap V_j=\{\}$ if $i\ne j$.
	The partitions of the vertex set induce a partitioning of the edge
	set $E=\cup_{i=0}^{P-1} E_i$ 
	such that for each $(v,w)\in E$: $(v,w)\in E_i$ if $w\in V_i$.
	The graph partitions are $G_i=(V,E_i)$, where any vertex can
	appear as a source vertex (hence the vertex set is $V$)
	but the set of destinations is restricted to $V_i$.
	The VEBO optimization criteria are:
	\begin{itemize}
		\item minimize $max_{i=0}^{P-1} |E_i| - min_{i=0}^{P-1} |E_i|$ (\emph{edge balance})
		\item minimize $max_{i=0}^{P-1} |V_i| - min_{i=0}^{P-1} |V_i|$ (\emph{vertex balance}).
	\end{itemize}
	These criteria address the worst-case spread of load per partition.
	Alternatively, criteria based on variation could be formulated which could
	potentially assess load imbalance with more precision. However,
	we will demonstrate that the worst-case spread of load is limited
	to 1 edge and 1 vertex. As such, these criteria are appropriate.
	
	
	\begin{algorithm}[!t]
		\Input{Graph $G=(V,E)$; number of partitions $P$}
		\Output{Reordered sequence numbers $S[v]\in 0,\ldots,|V|-1$ for $v\in V$
			and partition end points $u[p]\in 0,\ldots,|V|-1$ for $p\in 0,\ldots,P-1$}
		
		Let w[P] = \{ 0 \}\tcp*{tracking edge count in each partition}
		Let u[P] = \{ 0 \}\tcp*{tracking vertex count in each partition}
		Let a[V] = \{ 0 \}\tcp*{assigned partition for each vertex}
		Consider the list $v_{(0)},v_{(1)},\ldots,v_{(n)}$ of vertices sorted by decreasing in-degree, i.e., $\{v_{(0)},v_{(1)},\ldots,v_{(n)}\}=V$ and
		$deg_{in}(v_{(i)})\ge deg_{in}(v_{(j)})$ when $i>j$\;
		Let n = $|V|$\tcp*{the number of vertices}
		Let m = $|\{v\in V: deg_{in}(v)>0\}|$\tcp*{the number of vertices}
		\tcp*{with non-zero degree}
		
		\tcp{Phase 1. Assign vertices with non-zero degree}
		\For{$t\leftarrow 0$ \KwTo $m-1$}{
			Let $v = v_{(t)}$\;
			Let $p = \argmin_{i=0,\ldots,P-1} w[i]$\;
			Let $a[v] = p$\tcp*{assign $v$ to partition $p$}
			Increase $w[p]$ by $deg_{in}(v)$\tcp*{update edge count}
			Increase $u[p]$ by 1\tcp*{update vertex count}
		}
		
		\tcp{Phase 2. Assign vertices with zero degree}
		\For{$t\leftarrow m$ \KwTo $n-1$}{
			Let $v = v_{(t)}$\;
			Let $p = \argmin_{i=0,\ldots,P-1} u[i]$\;
			Let $a[v] = p$\tcp*{assign $v$ to partition $p$}
			Increase $u[p]$ by 1\tcp*{update vertex count}
		}
		\tcp{Phase 3. Calculate new sequence numbers}
		Let $s[0] = 0$\;
		\For{$p\leftarrow 1$ \KwTo $P-1$}{
			Let $s[p] = s[p-1] + u[p-1]$\;
		}
		\For{$t\leftarrow 0$ \KwTo $n-1$}{
			Let $v = v_{(t)}$\;
			Let $S[v] = s[a[v]]$\tcp*{determine sequence number for $v$}
			Increment $s[a[v]]$ by 1\;
		}
		\caption{The VEBO reordering algorithm}
		\label{algo:re}
	\end{algorithm}
	
	\subsection{Algorithm Description}
	The core idea behind VEBO is to perform \emph{vertex reordering}:
	each vertex is assigned a new sequence number in the range $0,\cdots,n-1$
	in a way that enables
	Algorithm~\ref{algo:part} to generate optimal load balance.
	As such, vertex reordering precedes partitioning (Figure~\ref{fig:steps}).
	We follow an approach similar to the multi-processor job scheduling
	heuristic~\cite{graham:69:mp}:
	place a set of objects in order of decreasing size,
	for each object selecting the least-loaded partition.
	In our case, however, we 
	adapt the algorithm to balance both the number
	of objects (vertices) and their size (degree).
	
	The VEBO algorithm (Algorithm~\ref{algo:re}) consists of three
	phases:
	In the first phase, VEBO places vertices with non-zero
	in-degree
	in order of decreasing in-degree.\footnote{From here on, we will refer to in-degree as ``degree'' for brevity.}
	This achieves a near-equal edge count in each partition. We will show that
	edge imbalance is just 1 edge when the size of the placed objects
	follows a power-law distribution.
	In the second phase, zero-degree vertices are placed.
        We observe that real-world graphs may have many zero-degree vertices.
	These vertices do not affect edge balance.
	If any vertex imbalance is introduced during the first phase,
	the vertex imbalance
	is corrected by placement of the zero-degree vertices.
	The third phase reorders the vertices. It assigns new sequence numbers
	to the vertices such that each partition consists of a continuous
	sequence of vertices. This is important to retain spatial locality
	and NUMA locality
	during graph processing~\cite{kyrola2012graphchi,zhang2015numa,sun:17:icpp}.
	
	\subsection{Example}
	Figure~\ref{fig:vebo_example} shows a short example.
	VEBO first sorts the vertices by decreasing in-degree.
	Second, it assigns vertices one by one to the partition
	that has the fewest incoming
	edges among all of the vertices already assigned to it.
	When all vertices have been placed, the vertices are assigned new IDs
	such that each partition spans a range of consecutive vertex IDs.
	This actual reordering is beneficial for spatial locality.
	Finally, a new graph representation is generated using the new vertex IDs.
	Figure~\ref{fig:vebo_example}b) shows how the graph is partitioned.
	Each partition has 7 incoming edges and 3 destination vertices. 
	
	\begin{figure}[t!]
		\centering
		\includegraphics[width=\columnwidth,clip]{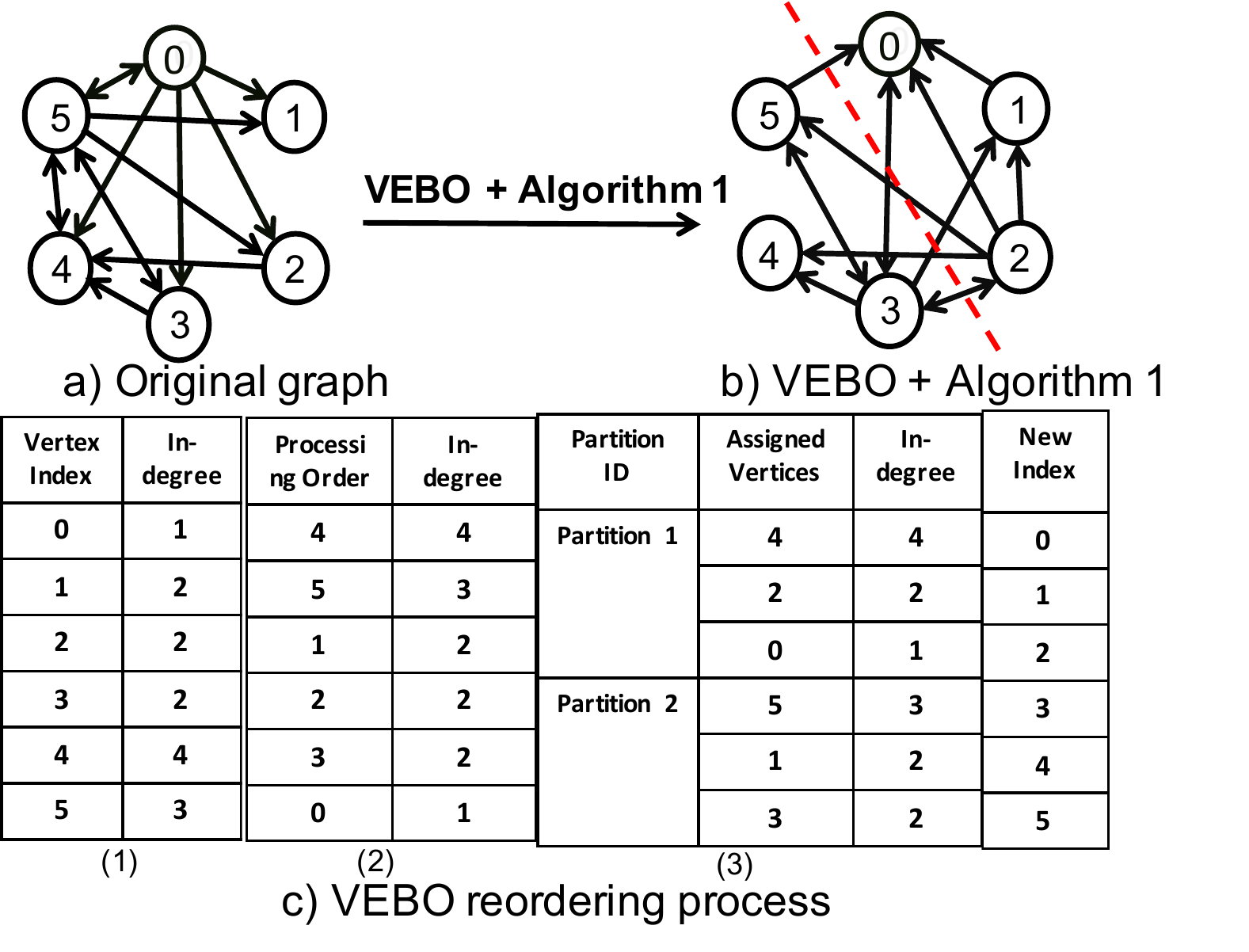}
		\caption{Illustration of VEBO on a 6-vertex graph}
		\label{fig:vebo_example}
		\vspace*{-3mm}
	\end{figure}
	
	\subsection{Analysis}
	To support analysis of the algorithm, we introduce some auxiliary definitions.
	Let $v_{(0)},v_{(1)},\ldots,v_{(n-1)}$ be the sequence that holds all
	vertices in $V$ in order of decreasing degree.
	Assume the algorithm goes through $n=|V|$ steps to place the vertices.
	Step $t>0$ corresponds to placing vertex $v_{(t)}$.
	The algorithm terminates when $t=|V|$.
	
	Let $w_p(t)$ be the number of edges assigned to partition $p$
	before step $t$, i.e., before placing vertex $v_{(t)}$.
	The initial situation is $w_p(0)=0$ for all $p$.
	%
	The maximum weight before step $t$ is:
	$\omega(t)=\max_{j=1,\ldots,P} w_j(t)$
	%
	The minimum weight before step $t$ is:
	$\mu(t)=\min_{j=1,\ldots,P} w_j(t)$
	%
	The edge imbalance before step $t$ is:
	$\Delta(t)=\omega(t)-\mu(t)$
	An optimal edge placement is achieved when $\Delta(n)\le 1$, i.e.,
	the number of edges in each partition differs by at most 1.
	$\Delta(n)=0$ can occur only when the number of partitions divides
	into the number of edges.
	%
	The vertex imbalance before step $t$ is:
	$\delta(t) = \max_{j=1,\ldots,P} u_j(t) - \min_{j=1,\ldots,P} u_j(t)$
	where $u_p(t)$ is the number of vertices assigned to partition $p$
	before step $t$.
	
	We first prove the following Lemma that bounds the edge imbalance
	throughout the placement of vertices.
	As a short-hand, let $d_{(t)}=deg_{in}(v_{(t)})$ for $t=0,\ldots,n-1$.
	Clearly, $d_{(i)}\ge d_{(j)}$ if $i<j$ and $d_{(0)}=N-1$.
	\begin{lemma}
		\label{lm:1}
		When placing a vertex $v_{(t)}$ with degree $d_{(t)}$ when
		the edge weight is $w_j(t)$ for $j\in 1,\ldots,P$,
		one of following cases can occur:
		\begin{eqnarray}
		\label{case1}
		\left.
		\begin{array}{ll}
		\Delta(t+1) \le \Delta(t)\\
		\omega(t+1) = \omega(t)
		\end{array}
		\right\} & \mbox{if} & d_{(t)} \le \Delta(t) \\
		\label{case2}
		\left.
		\begin{array}{ll}
		\Delta(t+1) \le d_{(t)} \\
		\omega(t+1) > \omega(t)
		\end{array}
		\right\} & \mbox{if} & d_{(t)} > \Delta(t)
		\label{lm:1:2}
		\end{eqnarray}
	\end{lemma}
	\begin{proof}
		We sketch the proof for brevity.
		Vertex $v_{(t)}$ is placed on the partition $p$
		with minimal $w_p(t)$, i.e., $w_p(t)=\mu(t)$.
		The edge count of $p$ increases to $\mu(t)+d_{(t)}$.
		Two cases arise depending on whether
		the maximum load on a parition is raised
		($\mu(t)+d_{(t)}>\omega(t)$ and thus $d_{(t)}>\Delta(t)$) or not
		($d_{(t)}\le\Delta(t)$). The Lemma follows by elaboration of the
		definitions of $\Delta(t)$ and $\omega(t)$.
	\end{proof}
	Intuitively, by placing more edges we either strive towards balancing
	the edge counts (case~\ref{case1}),
	or we are so close to load balance that placing the next vertex must
	increase the load imbalance (case~\ref{case2}).
	Importantly, in the latter case, the load imbalance is bounded
	by the degree of the last vertex placed.
	As we process vertices in order of decreasing degree, the edge imbalance
	reduces throughout the algorithm.
	
	\begin{theorem}[Edge balance]
		\label{th:1}
		Assume a graph $G=(V,E)$ and a number of partitions $P$.
		Let $n=|V|$ be the number of vertices
		and let $m$ be the number of vertices with non-zero degree.
		Assume that the degree distribution of the graph follows
		a Zipf distribution with $N$ distinct ranks and scale factor $s>0$. 
		Let $|E|$ be the number of edges.
		Assume that $P$ is constrained by $|E|\ge N\:(P-1)$ and that $P<N$.
		Then, on completion of VEBO (Algorithm~\ref{algo:re}),
		\begin{equation}
		\Delta(n)\le 1
		\end{equation}
	\end{theorem}
	\begin{proof}
		We sketch the proof.
		It builds on the observation that
		the final edge imbalance $\Delta(n)$ is at most 1
		if the maximum weight $\omega(t)$ can be increased by placing a degree-1 vertex
		(Lemma~\ref{lm:1}, Eq~\ref{lm:1:2}).
		Degree-1 vertices are abundant under the assumption of the Zipf distribution.
		As such, one can show that at
		any step $t$ where $d_{(t)}>\Delta(t)$, there are at least $(P-1) \Delta(t)$
		edges among the unplaced vertices, i.e.,
		$\sum_{i=t^{'}}^{m-1} d_{(i)}\ge (P-1)\Delta(t)$.
		This relation follows from the assumption $|E|\ge P(N-1)$
		and the recurrence
		$\sum_{i=1}^{N_1} (i-1)\:i^{-s}/N_1\ge \sum_{i=1}^{N} (i-1)\:i^{-s}/N$
		if $N_1<N$.
	\end{proof}

	The condition $|E|\ge N(P-1)$ can be understood as follows:
	As the highest degree is $N-1$, and all edges pointing to this vertex
	are placed in the same partition, at least one partition will have
	$N-1$ edges. In order to have edge balance, the other $P-1$ partitions
	may have at most $N$ edges. As such, $|E|\ge(N-1)+(P-1)(N-2)=P(N-2)+1$ is a necessary
	condition for edge balance.
	The theorem makes only a slightly stronger requirement.
	
	\begin{theorem}[Vertex balance]
		\label{th:2}
		Assume a graph $G=(V,E)$ and a number of partitions $P$
		as required for Theorem~\ref{th:1}.
		Assume that $n\ge N\:H_{N,s}$.
		Under these conditions, upon
		placement of vertices $v_{(0)},v_{(1)},\ldots,v_{(m-1)}$
		by Algorithm~\ref{algo:re}, it holds that $\delta(m)<N/P$
		and upon placement of all vertices
		by Algorithm~\ref{algo:re}, it holds that $\delta(n)\le 1$.
	\end{theorem}
	\begin{proof}
		The proof first shows that the vertex imbalance
		is bound as $\delta(m)<N/P$ after placing
		vertices with non-zero degree. Then, vertex balance can be
		achieved if at least $\delta(m)\:(P-1)$ zero-degree vertices
		are available, which follows from the properties of the
		degree distribution.
	\end{proof}

	The condition $n\ge N\:H_{N,s}$ in Theorem~\ref{th:2}
	dictates that there should be a
	sufficiently large number of vertices in the graph.
	The number of vertices is independent of the parameters $N$
	and $s$ which determine the shape of the degree distribution.
	This constraint is not stringent, e.g., if $s=1$, then
	the requirement is $n\ge 2\:N$.

	\begin {table*}[t!]
	\caption{Characterization of real-world and synthetic graphs used in experiments.}
	\label{tab:graphdetails}
	\small
	\scriptsize
	\centering
	\begin{tabular}{|l||r|r|r|r|r|r|r|r|l|}
		\hline
		& & & Max. & \% vertices with & \% vertices with & & & \\
		Graph & Vertices & Edges & Degree & zero in-degree  & zero out-degree  & $\delta(n)$ & $\Delta(n)$ &Type\\ \hline
		\hline
		Twitter~\cite{kwak2010twitter} & 41.7M & 1.467B & 770,155 &14\%&4\%&1&1& directed \\ \hline
		Friendster~\cite{SNAP}&	125M&	1.81B & 4,223&48\%& 37\%&1&1& directed \\\hline
		Orkut~\cite{SNAP}&	3.07M	&234M & 33,313 &0\%(186)&0\%(186)&2&1& undirected \\\hline
		LiveJournal~\cite{SNAP}	& 4.85M&	69.0M & 13,906 &7\%&21\%&1&1& directed \\\hline
		Yahoo\_mem~\cite{SNAP}&	1.64M&	30.4M & 5,429&0\%(0)& 0\%(0)&9&3& undirected \\\hline
		USAroad~\cite{zhang2015numa} & 23.9M & 58M & 9 &0\%(1)&0\%(1)&1&1& undirected \\ \hline
		Powerlaw ($\alpha=2$) \footnotemark & 100M & 294M &132,423  &0\%(99)&0\%(99)&1&1& undirected \\ \hline
		RMAT27~\cite{Ligra} & 134M & 1.342B & 812,983 &69\%&69\%&1&1& directed \\\hline
	\end{tabular}
\end{table*}

VEBO is applied to 8 graphs with a selection of
synthetic and real-world graphs (Table~\ref{tab:graphdetails}).
Seven graphs are power-law graphs. The USARoad graph represents
a road network and has nearly constant degree.
VEBO calculates an optimal
vertex- and edge-balanced placement with up to
384 partitions, i.e., $\Delta(n)=1$ and $\delta(n)=1$, for 6 graphs,
including the USAroad graph which is not scale-free.
For the remaining 2 graphs the largest discrepancy between
partitions is less than 10 edges or vertices out of millions.

We have introduced several constraints in the proofs, namely
$P\le |E|/(N-1)$ and $|V|>N\:H_{N,s}$. It can be observed from
Table~\ref{tab:graphdetails} that these constraints pose
no practical limits.

Theorem~\ref{th:2} is built on the premise that there are a substantial
number of vertices with zero in-degree. This happens frequently
in directed scale-free graphs, but less frequently in undirected
graphs (Table~\ref{tab:graphdetails}).
Nonetheless, the degree distribution in scale-free graphs is such
that VEBO achieves both edge and vertex balance in practice.


Algorithm~\ref{algo:re} has one drawback: vertices with
consecutive IDs in the original graph tend to be dispersed
across partitions, breaking any spatial locality that may be present in the
graph.
To retain spatial locality, we adjust phases 2 and 3 of the algorithm
to (i)~calculate how many vertices with the same degree are placed on each partition;
(ii)~assign blocks of consecutive vertices to the same partition.
We use this modification for the remainder of this paper.
\footnotetext{\url{https://github.com/snap-stanford/snap/tree/master/examples/graphgen}}
\subsection{Time Complexity}
Algorithm~\ref{algo:re} consists of three consecutive loops iterating
over the vertices. In the first two loops, all statements take a constant
number of time steps, except for the $\argmin$ operation which takes
$\mathcal{O}{(\log P)}$ steps when implemented using a min-heap.
As such these loops take time $\mathcal{O}(|V|\log P)$.
Sorting the vertices by degree can be achieved in 
$\mathcal{O}(|V|)$ using knowledge of the degrees ($0,\ldots,|V|-1$),
similar to radix sort.
The total time complexity of the VEBO algorithm is thus
$\mathcal{O}(|V|\log P)$.

Previously studied vertex reordering algorithms
are computationally more complex.
The algorithm presented by Li~\emph{et al}~\cite{li:17:spac}
has polynomial time complexity in $|V|$.
Gorder~\cite{wei:16:gorder} takes
$\mathcal{O}(\sum_{v\in V} (deg_{out}(v))^2)$ steps
where $deg_{out}(v)$ is the out-degree of vertex $v$.
The time complexity of RCM
is $\mathcal{O}(N \log N |V|)$ where $N$ is the
highest vertex degree.\footnote{\url{http://www.boost.org/doc/libs/1_66_0/libs/graph/doc/cuthill_mckee_ordering.html}}
The difference in complexity is evident as these algorithms
solve a more complex problem.



\begin {table}[t!]
\caption{Graph algorithms and their characteristics.
Traversal direction (B=backward, F=forward) is used by Ligra and Polymer.
Vertex (V) or Edge (E) orientation is used by GraphGrind. 
Frontiers (F) are dense (d), medium-dense (m) or sparse (s).}
\label{tab:app}
\scriptsize
\centering
\begin{tabular}{|p{0.75cm}||p{.45\columnwidth}|c|c|c|}
\hline
\textbf{Code} & \textbf{Description}
& \textbf{B/F}
& \textbf{V/E}
& \textbf{F}
\\ \hline
\hline
BC &
betweenness-centrality~\cite{Ligra}
& B
& V 
& m/s
\\ \hline
CC &
connected components using label propagation~\cite{Ligra}
& B
& E 
& d/m/s
\\ \hline
PR &
Page-Rank using power method (10 iterations)~\cite{page:99:pagerank}
& B
& E
& d
\\ \hline
BFS &
breadth-first search~\cite{Ligra}
& B
& V 
& m/s
\\ \hline
PRD &
optimized Page-Rank with delta-updates~\cite{Ligra}
& F
& E 
& d/m/s
\\ \hline
SPMV &
sparse matrix-vector multiplication (1 iteration)
& F
& E 
& d
\\ \hline
BF &
single-source shortest path (Bellman-Ford)~\cite{Ligra}
& F
& V 
& d/m/s
\\ \hline
BP &
Bayesian belief propagation~\cite{zhang2015numa} (10 iterations)
& F
& E 
& d
\\ \hline
\end{tabular}
\end{table}

\section{Evaluation Methodology}
\label{sec:evalm}
%
We experimentally evaluate VEBO and compare it to two state-of-the-art
vertex reordering algorithms:
RCM~\cite{george:94:sparse}
and Gorder\footnote{\url{https://github.com/datourat/Gorder}}~\cite{wei:16:gorder}. 
The RCM algorithm aims to reduce the bandwidth of a sparse matrix
and is known to work well for applications in numerical analysis~\cite{george:94:sparse}.
Gorder aims to improve temporal locality in graph analytics.

We use three shared memory graph processing systems:
Ligra~\cite{Ligra}, Polymer~\cite{zhang2015numa}
and GraphGrind~\cite{sun:17:ics,sun:17:icpp}.
They model graph analytics as iterative algorithms
where a set of \emph{active vertices}, known as \emph{frontier},
is processed on each iteration. Vertices become active when the
values calculated for them are updated.
The three systems use two
key functions:
the \emph{edgemap} function applies an operation to all edges
whose source vertex is active,
while the \emph{vertexmap} function applies an operation to all
active vertices.
All systems implement the direction reversal heuristic~\cite{beamer:12:bfs}
and dynamically adjust the frontier data structures depending on the
frontier size. 

The frontier varies
during the execution and affects the best way to traverse the graph.
Frontier density is measured as the number of active vertices
and active edges divided by the number of edges.
There are three types of frontiers, dense, medium-dense and sparse in GraphGrind.
Table~\ref{tab:app} shows these properties for some
commonly used graph algorithms.

For the purpose of this work,
the key differences between these systems is in the scheduling of
parallel work and memory locality optimization.
Ligra expresses parallelism using Cilk~\cite{frigo:98:cilk5}, which is fully
dynamically scheduled.
Ligra contains no specific optimizations for memory locality.
Polymer expresses parallelism using POSIX threads and uses static scheduling.
GraphGrind uses a mixture of static and dynamic scheduling.
Static scheduling is used
to bind partitions to NUMA sockets, while
dynamic scheduling is used internally in a socket to distribute work
across threads.
GraphGrind increases temporal locality
by creating more partitions than threads~\cite{sun:17:icpp}
and by traversing edges in Hilbert order~\cite{murray:13:naiad}.



We evaluate the performance benefits of VEBO experimentally on
a 4-socket 2.6GHz Intel Xeon E7-4860 v2 machine, totaling 48 threads
(we disregard hyperthreading due to its inconsistent impact on performance)
and 256$\:$GB of DRAM.
We compile all codes using the Clang compiler with Cilk support.
We evaluate 8 graph analysis algorithms (Table~\ref{tab:app}),
using 8 widely used graph data sets (Table~\ref{tab:graphdetails}).
Our evaluation is missing
results for Betweenness Centrality (BC) on Polymer
as Polymer does not provide an implementation for it.
We exclusively present results using 48 threads and present
averages over 20 executions.

We generate 4 partitions with VEBO for Polymer, as it uses one partition
per NUMA node and we generate 384 partitions with VEBO for GraphGrind,
which is their recommended~\cite{sun:17:icpp}.
Ligra does not partition graphs.


\section{Experimental Evaluation}
\label{sec:eval}


\subsection{Performance Overview}
We evaluated the execution time achieved with
the original graph, RCM, Gorder and VEBO using each of
the Ligra, Polymer and GraphGrind processing systems (Table~\ref{tab:order}).
GraphGrind reorders edges in Hilbert order when using
the COO~\cite{sun:17:icpp} for all results except
for VEBO. While VEBO works well in conjunction with Hilbert order,
we found it works even better when storing edges in CSR order.
The reason behind this is explained in Section~\ref{sec:sfc}.
We will first discuss the results for the power-law graphs; the
USARoad graph is discussed separately.

The best
vertex order for Ligra varies between algorithms and between graphs.
Sometimes, all ``optimized'' vertex orders result in worse performance
compared to using the original graph (e.g., BFS on the Twitter graph).

Gorder assumes that all vertices
and edges are active, as it does not know how the frontier will evolve
during computation.
As such, algorithms with dense frontiers benefit most: PR, PRD, SPMV, BF and BP.
Gorder often results in a slowdown
for algorithms with sparse frontiers, such as CC, BC and BFS. 

\begin{table*}[t!]
  \caption{Runtime in seconds of Ligra, Polymer and GraphGrind using original graph, VEBO, Gorder and RCM.
    The fastest results for each combination of algorithm, graph and framework are indicated in bold-face. Slowdowns over original graphs are indicated in italics.}
	\scriptsize
	\centering
        \addtolength{\tabcolsep}{-1pt}
	\begin{tabular}{|c|c||r|r|r|r|}
    \hline
     \multicolumn{2}{|c||}{}&\multicolumn{4}{|c|}{\textbf{Ligra}} \\
    \hline
    \textbf{Graph} &  \textbf{Algo.} 
    & \multicolumn{1}{c|}{\textbf{Orig.}}    & \multicolumn{1}{c|}{\textbf{RCM}}
    & \multicolumn{1}{c|}{\textbf{Gorder}} & \multicolumn{1}{c|}{\textbf{VEBO}}\\
    \hline
    \hline
    
    \parbox[t]{2mm}{\multirow{8}{*}{\rotatebox[origin=c]{90}{Twitter}}}
    & CC & \textbf{3.132} & \textit{4.242} & \textit{3.217} &\textit{3.596}  \\
    & BC & 2.798 & 2.492  & \textit{4.248} & \textbf{1.846}\\
    & PR & 22.143 &21.702& 20.142 &  \textbf{19.509}  \\
    & BFS & \textbf{0.347} & \textit{0.699} & \textit{0.581} &  \textit{0.472}\\
    & PRD& 35.110 &\textit{43.617} & \textbf{32.695} & \textit{38.589}  \\
    & SPMV & 4.311 &2.281  & \textbf{1.602} & 1.847 \\
    & BF & 4.255 & \textit{4.649} & \textit{9.394} &\textbf{3.154}  \\
    & BP & 68.767 &\textit{105.223} & \textbf{64.242} & \textit{90.439}  \\
    \hline  
    
    \parbox[t]{2mm}{\multirow{8}{*}{\rotatebox[origin=c]{90}{Friendster}}}
    & CC & 7.031 & \textbf{6.636} & \textit{13.517} & 6.831 \\
    & BC & 5.499 & 3.212  & 5.038 &\textbf{3.170}  \\
    & PR & 47.233 & 37.113  & \textbf{29.532} &36.587  \\
   & BFS & 1.441 & \textbf{0.808} & \textit{1.620} & 1.073 \\
    & PRD& 65.886 &64.210  & \textbf{41.021} & 58.640 \\
    & SPMV & 10.112 & 3.730  & 7.300 &\textbf{3.535} \\
    & BF& 8.884 &  6.161& \textbf{6.136}  &8.787 \\
    & BP& 151.742 & 132.865 & \textbf{99.292} &125.785  \\
    \hline  
    
    \parbox[t]{2mm}{\multirow{8}{*}{\rotatebox[origin=c]{90}{RMAT27}}}
    & CC & 3.544 & \textit{4.690} & \textit{6.779} & \textbf{3.181}  \\
    & BC & 2.567 & \textit{3.544} & \textit{5.991} &\textbf{2.250}  \\
    & PR & 29.965 &  27.012 & \textit{43.379} &\textbf{21.266}\\
    & BFS& \textbf{0.493} & \textit{0.827} & \textit{0.701} &  \textit{0.647}\\
    & PRD& 17.688 &\textit{20.243} & \textit{20.962} &  \textbf{16.726} \\
    & SPMV & 3.883 & 2.657 & 3.188 & \textbf{2.340} \\
    & BF& 3.718 & 3.320 & \textbf{2.196} & 2.735 \\
    & BP& 75.336 &\textit{76.805}& \textbf{70.729} &  73.196  \\
    \hline  
    
    \parbox[t]{2mm}{\multirow{8}{*}{\rotatebox[origin=c]{90}{PowerLaw}}}
    & CC & 5.281 & 4.940 & \textbf{3.302} &  3.622 \\
    & BC & 2.902 & 2.667 & \textbf{2.331} & 2.589\\
    & PR &12.263 &  11.621  & 11.954 &\textbf{11.383}\\
    & BFS &1.010 & 0.983 & \textbf{0.861} &0.890  \\
    & PRD&22.698 & 21.688 & \textbf{20.821} & 21.947 \\
    & SPMV &1.224 & \textit{1.421}  & 1.174 &\textbf{1.172} \\
    & BF&7.143 & 6.412  & \textbf{6.212} &6.226 \\
    & BP& 23.474 &19.317 & \textbf{18.001} & 21.122  \\
    \hline  
    
    \parbox[t]{2mm}{\multirow{8}{*}{\rotatebox[origin=c]{90}{Orkut}}}
    & CC & 0.151 &\textit{0.172}  & \textit{0.157} &  \textbf{0.149}\\
    & BC & 0.177 & \textbf{0.141} &0.170  &0.160 \\
    & PR & 	3.184 & 2.569 & \textbf{1.899} &2.437  \\
    & BFS & 0.041 &  \textbf{0.038}& \textit{0.052} & 0.041 \\
    & PRD&\textbf{2.301} &\textit{5.074} & \textit{3.252} &\textit{4.502}   \\
    & SPMV & 0.431 &  0.116& \textbf{0.104} &0.202  \\
    & BF& 0.357 & \textit{0.386} & \textbf{0.279} &  \textit{0.402}\\
    & BP& 6.372 &\textit{12.643} & \textit{8.107} &  \textbf{5.398} \\
    \hline  
    
    \parbox[t]{2mm}{\multirow{8}{*}{\rotatebox[origin=c]{90}{LiveJournal}}}
    & CC & 0.133 & \textit{0.170}& \textbf{0.107} &  0.108 \\
    & BC & 0.194 & 0.140 & \textbf{0.128} &0.159  \\
    & PR & 1.219 & \textit{2.594}  & \textbf{0.675} &0.865 \\
    & BFS & 0.056 & 0.047 & 0.047 & \textbf{0.046} \\
    & PRD& 1.464 &\textit{4.747}  & \textbf{1.084} &  \textit{1.629}\\
    & SPMV & 0.171 & 0.119& 0.102 &  \textbf{0.054} \\
    & BF& 0.555 &0.356  & 0.464 &  \textbf{0.255}\\
    & BP & 3.522 & \textit{4.335} & \textbf{1.880} &  \textit{3.672}\\
    \hline 
    
    \parbox[t]{2mm}{\multirow{8}{*}{\rotatebox[origin=c]{90}{Yahoo\_mem}}}
    & CC & 0.080 &\textbf{0.036} & 0.045 &  0.039  \\
    & BC & 0.167 & \textbf{0.052}& 0.057 & 0.081 \\
    & PR & 0.684 &0.285 & \textbf{0.246} &  0.265  \\
    & BFS & 0.035 & 0.026 & 0.025 & \textbf{0.024} \\
    & PRD& 2.501 &\textit{2.590} & \textbf{1.607} &  2.489 \\
    & SPMV & 0.053 & \textbf{0.020} & 0.046 &0.022  \\
    & BF& 0.357 &\textbf{0.166}& 0.237 &  0.174  \\
    & BP& 2.372 & 1.659 & \textbf{1.235} &  1.682\\
    \hline 
    
    \parbox[t]{2mm}{\multirow{8}{*}{\rotatebox[origin=c]{90}{USAroad}}}
    & CC & 38.669 &\textit{54.119}   & \textbf{7.953} & 24.848\\
    & BC & \textbf{4.620} &\textit{4.783}   & \textit{4.964} & \textbf{4.655}\\
    & PR & \textbf{1.559} &\textit{1.855}& \textbf{1.559} &   \textit{1.957} \\
    & BFS & \textbf{1.621} &\textit{ 1.819}& \textit{1.937 }& \textit{1.699 }\\
    & PRD& \textbf{2.886} & \textit{ 2.982} & \textit{2.975 }&\textit{3.505}\\
    & SPMV & \textbf{0.120} &  \textit{0.135} & \textit{0.143 }&\textit{0.163 } \\
    & BF& \textbf{28.848} &\textit{ 29.175} & \textit{34.646 }&\textbf{29.013}  \\
    & BP& \textbf{1.730} & \textit{1.843} &\textit{ 1.785} &\textit{ 2.093 }\\
    \hline 
    \end{tabular}%
\begin{tabular}{||r|r|r|r|}
	\hline
	\multicolumn{4}{||c|}{\textbf{Polymer}} \\
	\hline
	\multicolumn{1}{||c|}{\textbf{Orig.}}    & \multicolumn{1}{c|}{\textbf{RCM}}
	& \multicolumn{1}{c|}{\textbf{Gorder}} & \multicolumn{1}{c|}{\textbf{VEBO}}\\
	\hline
	\hline
	
	2.708 & \textit{2.930}   & \textit{2.972}  & \textbf{2.443}\\
    & & & \\
    20.948 & \textit{22.271}  & \textit{25.934} & \textbf{17.603}\\
    0.323 & 0.321           & \textit{0.336} &  \textbf{0.296}\\
    29.151 & 27.670          & \textit{33.336} & \textbf{19.237}\\
    3.746 & 3.293         & 3.183          & \textbf{1.633}\\
    3.990 & \textit{9.325}  & \textit{11.246} & \textbf{3.036}\\
    57.310 & 50.398 & 53.365 & \textbf{44.366} \\
    \hline
    
	6.523 & \textit{6.775}   & \textit{10.035} & \textbf{5.136}\\
	    & & &  \\
	45.410 & \textit{46.573}  & 32.827          & \textbf{25.776}\\
	1.308 & 1.113           & \textit{2.214} &\textbf{1.003} \\
	50.331 & \textit{63.000} & \textit{56.370} & \textbf{37.997} \\
	8.124 & 6.337         & \textit{9.934} & \textbf{3.364} \\
	7.653 & 7.312           & 7.366           &  \textbf{7.004}\\
	96.113 & 85.475 & 75.392 & \textbf{65.365} \\
	\hline
	
	2.880 & 2.622            & 2.543           & \textbf{2.123}\\
	& & &  \\
	21.645 &  19.958          & 19.278          & \textbf{16.544}\\
	0.456 & 0.440           & 0.443          &\textbf{0.422} \\
	12.134 & 11.034          & \textit{14.366} & \textbf{8.922}  \\
	2.538 & 2.443         & 2.331          & \textbf{2.237}  \\
	3.090 & \textit{4.331}  & \textit{4.672}  & \textbf{2.557}  \\
	68.324 & 50.440 & 58.035 & \textbf{40.384}\\
	\hline
	
	4.063 & \textit{4.597}   & 3.002           & \textbf{2.994}\\
	& & & \\
	8.661  & \textit{9.503}   & \textit{8.917}  & \textbf{7.592} \\
	0.866 & \textit{0.891}  & \textit{0.880} &\textbf{0.806} \\
	16.335 & \textit{18.003} & \textit{16.687} & \textbf{14.887}  \\
	0.885 & \textit{1.003}& 0.807          & \textbf{0.766}\\
    6.032 & \textit{6.224}  & \textit{6.185}  & \textbf{5.996} \\
    17.624 & 16.023 & 16.753 & \textbf{15.336} \\
	\hline
	
	0.123 & \textit{0.151}   & \textit{0.208}  & \textbf{0.116}\\
    & & & \\
    2.003  &  1.893           & 1.888           & \textbf{1.630}\\
    0.039 & \textit{0.047}  & \textit{0.050} &\textbf{0.038} \\
    1.630  & \textit{1.888}  & \textit{1.806}  & \textbf{1.224}  \\
    0.288 & 0.263         & 0.277          & \textbf{0.166}\\
    0.345 & \textit{0.396}  & 0.331           & \textbf{0.313}\\
    5.652  & 4.199  & 4.522  & \textbf{4.038} \\
	\hline   
	
	0.133 & 0.129            & \textit{0.168}  & \textbf{0.107}\\
	& & &  \\
	1.080  & \textit{1.464}   & 1.010           & \textbf{0.808}\\
	0.055 & \textit{0.059}  & \textit{0.065} & \textbf{0.046}\\
	1.320  & \textit{1.555}  & \textit{1.503}  & \textbf{0.994}  \\
	0.134 & 0.122         & 0.112          & \textbf{0.049} \\
	0.468 & 0.577           & \textit{0.534}  & \textbf{0.251} \\
	2.376  & 2.114  & 1.886  & \textbf{1.774} \\
	\hline  
	
	0.049 & \textit{0.058}   & \textit{0.069}  & \textbf{0.038}\\
    & & &  \\
    0.274  & \textit{0.276}   & \textit{0.299}  & \textbf{0.234}  \\
    0.025 & \textit{0.026}  & 0.025          & \textbf{0.024} \\
    1.687  & \textit{2.932}  & 1.533           & \textbf{1.133} \\
    0.049 & 0.045         & 0.042          & \textbf{0.020} \\
    0.197 & 0.191           & \textit{0.312}  & \textbf{0.169}\\
    1.667  & 0.896  & 0.702  & \textbf{0.590}\\
	\hline

	36.877 & \textit{46.338 } & \textbf{7.834} &22.673 \\
    & & & \\
    \textbf{1.075} &\textit{ 1.703 } & \textit{1.382} &\textit{1.216} \\
    \textbf{1.588} & \textit{1.766 }& \textit{1.819 }& \textbf{1.593} \\
    \textbf{2.241} &\textit{ 2.733 }& \textit{2.584} &\textit{ 2.436} \\
    \textbf{0.079} &\textit{ 0.115} & \textit{0.132} & \textit{0.099} \\
    \textbf{24.067} & \textit{28.336} & \textit{31.334} & \textit{25.532} \\
    \textbf{1.343} &\textit{ 1.543}& \textit{1.391 }&\textit{ 1.422 } \\
    \hline
\end{tabular}%
	\begin{tabular}{||r|r|r|r|}
    \hline
    \multicolumn{4}{||c|}{\textbf{GraphGrind}} \\
    \hline
    \multicolumn{1}{||c|}{\textbf{Orig.}}    & \multicolumn{1}{c|}{\textbf{RCM}}
    & \multicolumn{1}{c|}{\textbf{Gorder}} & \multicolumn{1}{c|}{\textbf{VEBO}}\\
	\hline
	\hline
	 1.722 & \textit{2.250} & \textit{2.261} & \textbf{1.089}\\
	 1.478 & \textit{2.697} & \textit{4.188} & \textbf{1.342} \\
	 11.824 & \textit{11.979} & \textit{16.219} & \textbf{9.693}\\
	 0.245 & 0.234          & \textit{0.249} & \textbf{0.210}\\
	 15.102 &\textit{15.352} & \textit{19.613} & \textbf{10.258}\\
	 1.861 &  1.199         & 1.186          & \textbf{0.627} \\
	 3.877 & \textit{7.059} & \textit{9.907} & \textbf{2.735} \\
	 40.412 & 31.850          & 36.314 & \textbf{21.101} \\
	 \hline
	 
	 3.516 & \textit{3.530} & \textit{8.216} & \textbf{3.081} \\
	 3.428 & 2.947          & \textit{6.841} &  \textbf{2.753}\\
	 29.444 & \textit{29.981} & 27.569          & \textbf{15.306} \\
	 0.931 & 0.619          & \textit{1.890} & \textbf{0.513} \\
	 30.108 &\textit{36.666} & \textit{33.223} & \textbf{18.364} \\
	 3.511 &  2.051         & \textit{5.893} & \textbf{0.973} \\
	 7.105 & 6.255          & 6.264          & \textbf{6.131} \\
	 69.526 & 54.147          & 48.586 & \textbf{46.563} \\
	 \hline

	 2.656 & 2.511          & 2.198          & \textbf{1.060} \\
	 2.081 & \textit{3.207} & \textit{5.943} & \textbf{1.393} \\
	 19.250 & 15.395          & 15.602          & \textbf{7.489} \\
	 0.405 & 0.271          & 0.284          & \textbf{0.263} \\
	 9.002  & 8.601          & \textit{10.312} & \textbf{4.324}\\
	 1.814 &  1.506         & 1.364          & \textbf{0.594}\\
	 2.352 & \textit{3.247} & \textit{3.320} & \textbf{1.894} \\
	 40.092 & 29.702          & 31.763 & \textbf{16.866} \\
	 \hline	
	 
	 3.184 & \textit{3.662} & 2.997          & \textbf{1.458} \\
	 2.466 & \textit{3.212} & \textit{2.993} & \textbf{2.257} \\
	 6.936  & \textit{10.337} & \textit{7.334}  & \textbf{5.864} \\
	 0.716 & \textit{0.882} & \textit{0.843} & \textbf{0.659}  \\
	 13.155 &\textit{18.337} & \textit{13.342} & \textbf{11.547}\\
	 0.450 & \textit{0.835} & 0.440          & \textbf{0.391} \\
	 4.315 & \textit{4.863} & \textit{4.446} & \textbf{4.094} \\
	 12.843 & \textit{12.338} & 10.112 & \textbf{9.297} \\
	 \hline
	 
	 0.116 & \textit{0.117} & \textit{0.142} & \textbf{0.102} \\
	 0.172 & 0.161          & 0.162          &  \textbf{0.159}\\
	 1.337  & 1.288           & 1.275           & \textbf{1.022} \\
	 0.037 & \textit{0.040} & \textit{0.047} & \textbf{0.036}  \\
	 1.038  &\textit{1.101}  & \textit{1.095}  & \textbf{0.907} \\
	 0.199 & 0.069          & 0.089          & \textbf{0.061} \\
	 0.306 & \textit{0.377} & 0.258          & \textbf{0.224} \\
	 5.392  & 1.732           & 2.176  & \textbf{1.484} \\
	 \hline
	 
	 0.120 & 0.116          & \textit{0.149} & \textbf{0.083} \\
	 0.183 & 0.160          & \textit{0.210} & \textbf{0.154} \\
	 0.751  & \textit{1.283}  & 0.634           & \textbf{0.527} \\
	 0.046 & \textit{0.049} & \textit{0.057} & \textbf{0.045}  \\
	 1.061  &\textit{1.100}  & \textit{1.172}  & \textbf{0.729}\\
	 0.082 & 0.071          & 0.064          & \textbf{0.033} \\
	 0.305 & \textit{0.376} & \textit{0.365} & \textbf{0.247}  \\
	 1.190  & \textit{1.735}  & 0.750  & \textbf{0.618}\\
	 \hline	 
	 
	 0.042 & \textit{0.049} & \textit{0.057} & \textbf{0.035} \\
	 0.079 & \textit{0.088} & \textit{0.089} & \textbf{0.075} \\
	 0.226  & 0.215           & \textit{0.242} & \textbf{0.207} \\
	 0.024 & \textit{0.025} & \textbf{0.023} & \textbf{0.023} \\
	 0.676  &\textit{2.080}  & 0.665           & \textbf{0.617} \\
	 0.029 & 0.021          & 0.019          & \textbf{0.016} \\
	 0.188 & 0.166          & \textit{0.276} & \textbf{0.155}  \\
	 0.802  & 0.494           & 0.302  & \textbf{0.254} \\
	\hline		
	 
	 30.754 & \textit{41.188} & \textbf{7.709} &  19.829\\
     \textbf{3.892} &\textit{ 4.443}& \textit{4.198 }& \textit{ 3.954}\\
     \textbf{0.707} &\textit{1.330 } &\textit{ 1.282 }& \textit{0.960}\\
     \textbf{1.424} & \textit{1.586 } & \textit{1.728 }& \textit{1.541}\\
     \textbf{1.809} &\textit{2.044} &\textit{ 1.828 }& \textit{ 2.033 }\\
     \textbf{0.053} &\textit{0.103}& \textit{0.110} &  \textit{ 0.058 }\\
      \textbf{21.510} &\textit{ 26.334} & \textit{25.976 }& \textit{22.678} \\
     \textbf{1.245} & \textit{1.334} & \textit{1.258 }&\textit{ 1.268} \\
     \hline
\end{tabular}%
  \label{tab:order}%
\end{table*}%

Ligra does not explicitly partition the graph, yet VEBO results in speedups
for Ligra. This stems from the implicit partitioning applied
when iterating over the CSR or CSC representation using Cilk parallel for loops.
Cilk recursively splits the iteration range in two parts. Each part may
be executed by a distinct worker thread. Cilk loops are most efficient when
the recursive split results in balanced workloads in each part.
While each part has a comparable number of vertices by design of Cilk,
the number of edges traversed by each thread will vary depending on the
graph topology. VEBO improves load balance as every 384-th part of the
iteration range has identical vertex and edge counts.

Overall,
VEBO achieves an average speedup of 1.09x
over Ligra while Gorder has an average speedup of 1.17x over Ligra. 
We attribute this to the use of dynamic scheduling in Ligra,
which compensates for load imbalance, and the lack of locality optimization.
%

VEBO provides consistently best performance
on Polymer (1.41x speedup) and GraphGrind (1.65x speedup); 1.53x speedup
when traversing edges in Hilbert order).
These systems use static scheduling to bind code to NUMA domains,
which makes them more sensitive to load balance than Ligra.
Gorder and RCM are less effective than VEBO for Polymer and GraphGrind
as they try to optimize memory locality but not load balance.
\subsection{A Non-Power-Law Graph: USAroad}
\label{USAroad}
USAroad graph has a degree distribution
that is close to uniform (the maximum degree is 9), shows a distinct
behavior from the power-law graphs.
Execution times are increased for all algorithms
but CC (Table~\ref{tab:order}).
Further analysis has shown that the root problem is a significant degradation
of memory locality.
Road networks typically have strong locality
and can be
partitioned in such a way that there are many internal vertices
and few external vertices, i.e., few vertices have
edges shared with vertices in other partitions~\cite{chen2015optimal}.
VEBO is agnostic of this structure of the graph and thus breaks the locality.

A curious exception is Connected Components (CC, using
label propagation).
Synchronous algorithms propogate only data
calculated in the previous iteration.
For CC, an asynchronous~\cite{low:12:graphlab} 
implementation is correct
and results in an \emph{accelerated propagation} of labels:
labels determined during one iteration of the algorithm
are propagated to other vertices during the same iteration~\cite{Ligra}.
Graph reordering seems to amplify accelerated propagation.
This reduces the number of medium-dense iterations of the algorithm
and explains the speedup.

\begin{figure*}[t]
	\centering	
	\includegraphics[width=.6\columnwidth,clip]{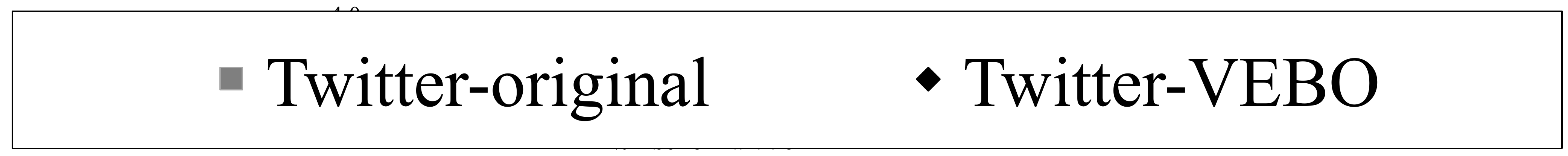}\\[-3mm]
	\subfloat[PR -- time]{
		\includegraphics[width=.19\textwidth,clip]{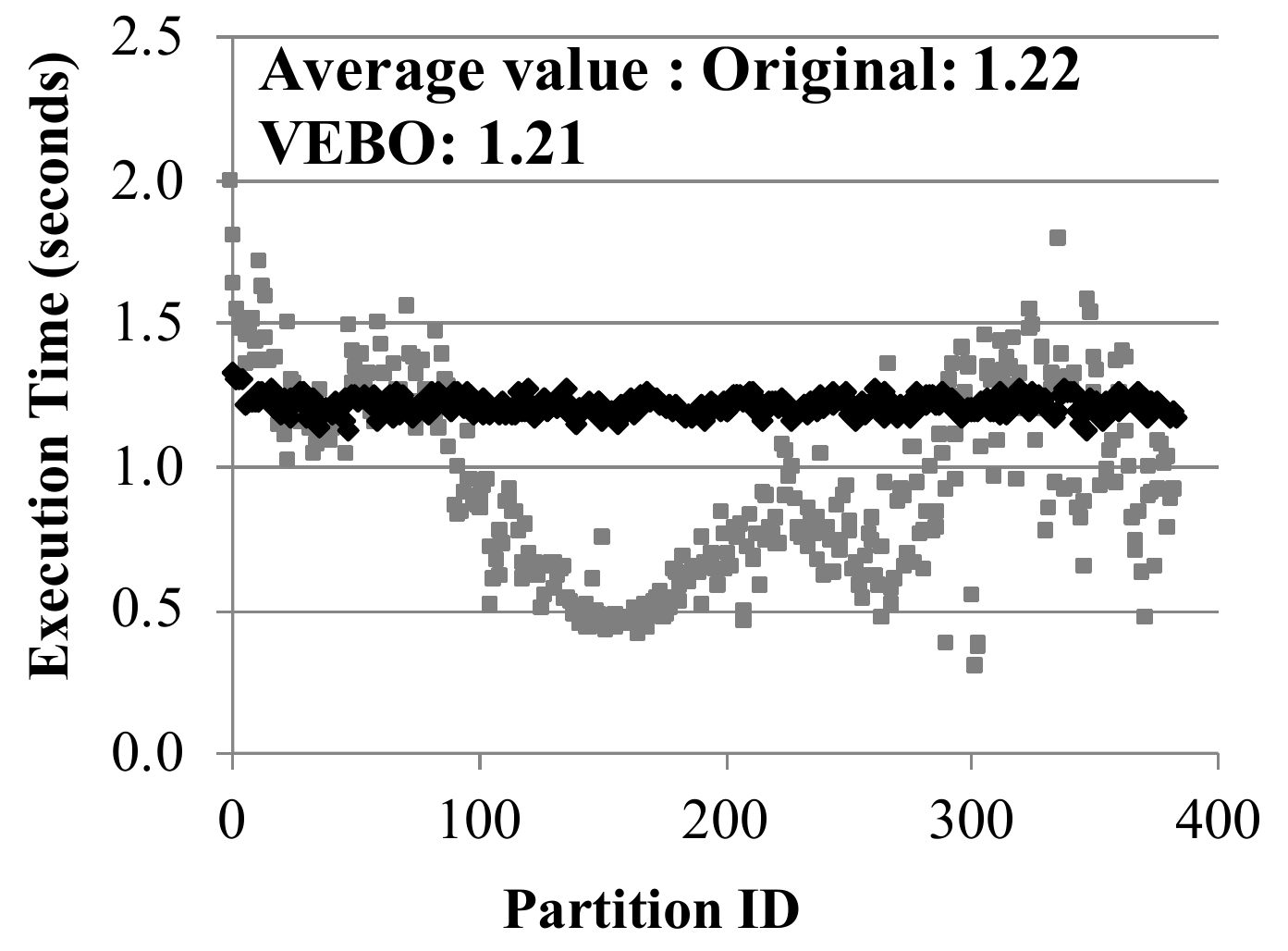}
		\label{fig:pr-time}
		\vspace*{-8mm}
	}
	\subfloat[PR -- local misses]{
		\hspace*{-4mm}
		\includegraphics[width=.19\textwidth,clip]{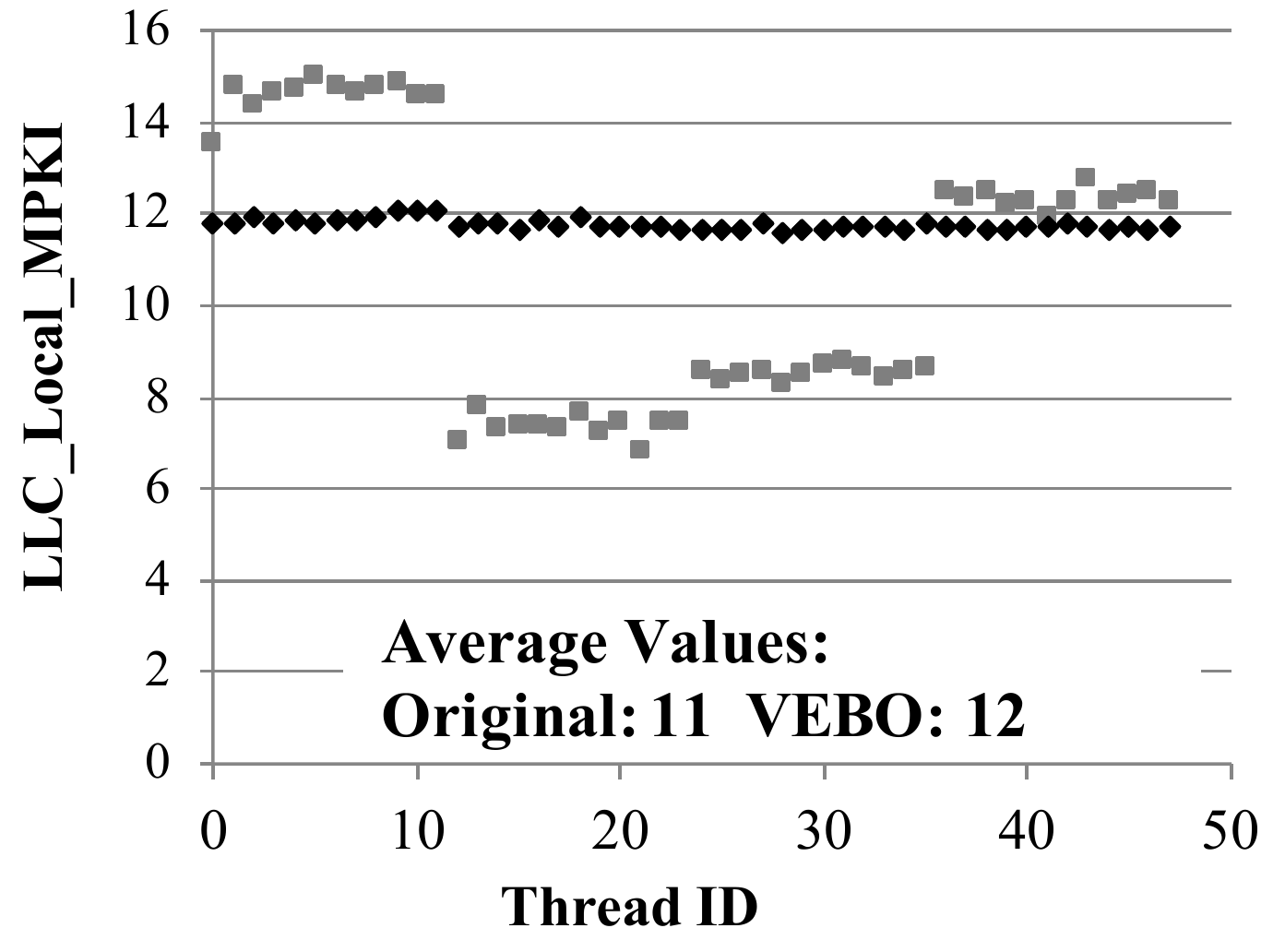}
		\label{fig:pr-local}
		\vspace*{-8mm}
	}
	\subfloat[PR -- remote misses]{
		\hspace*{-4mm}
		\includegraphics[width=.19\textwidth,clip]{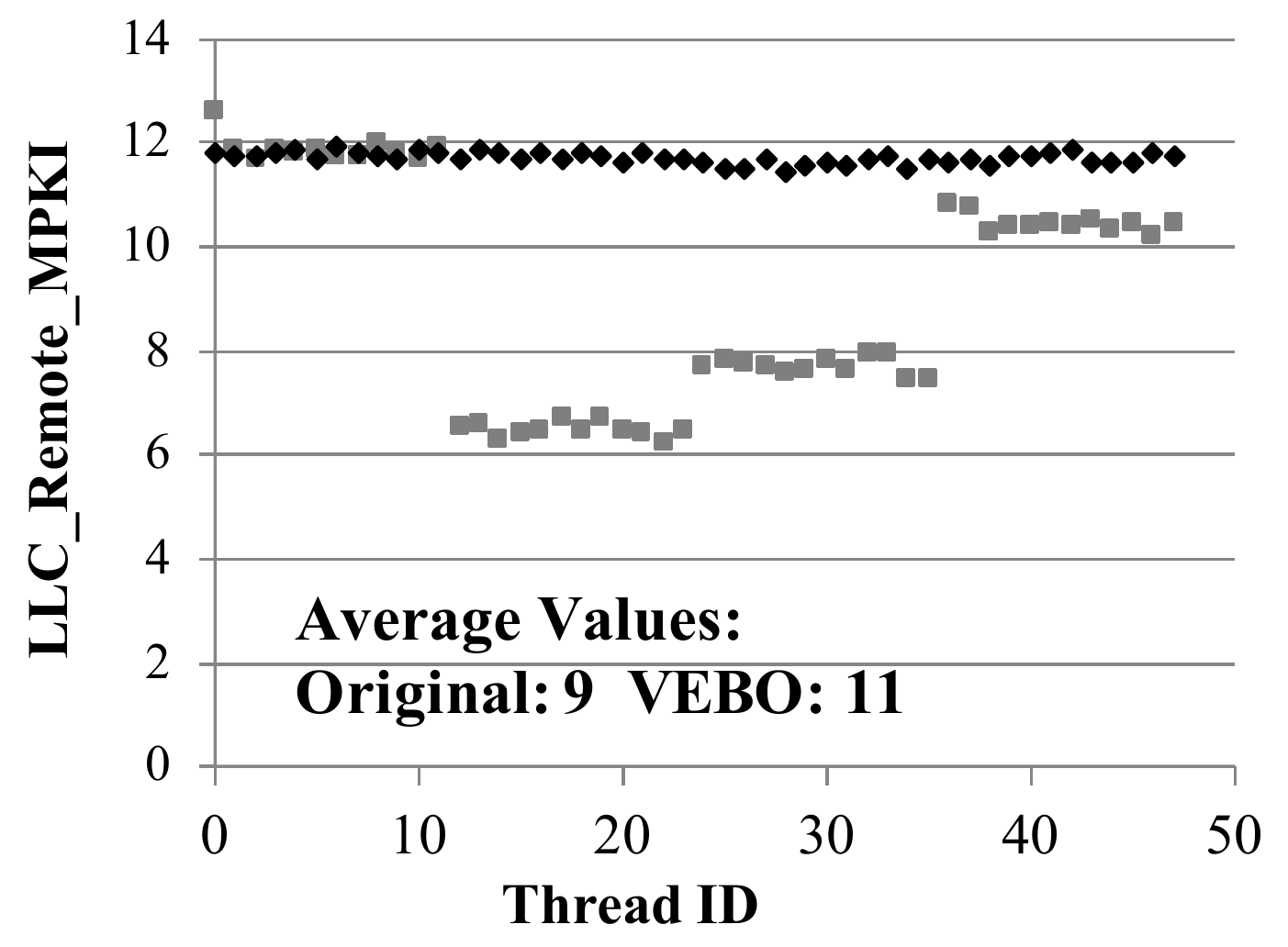}
		\label{fig:pr-remote}
		\vspace*{-8mm}
	}
	\subfloat[PR -- TLB]{
		\hspace*{-4mm}
		\includegraphics[width=.19\textwidth,clip]{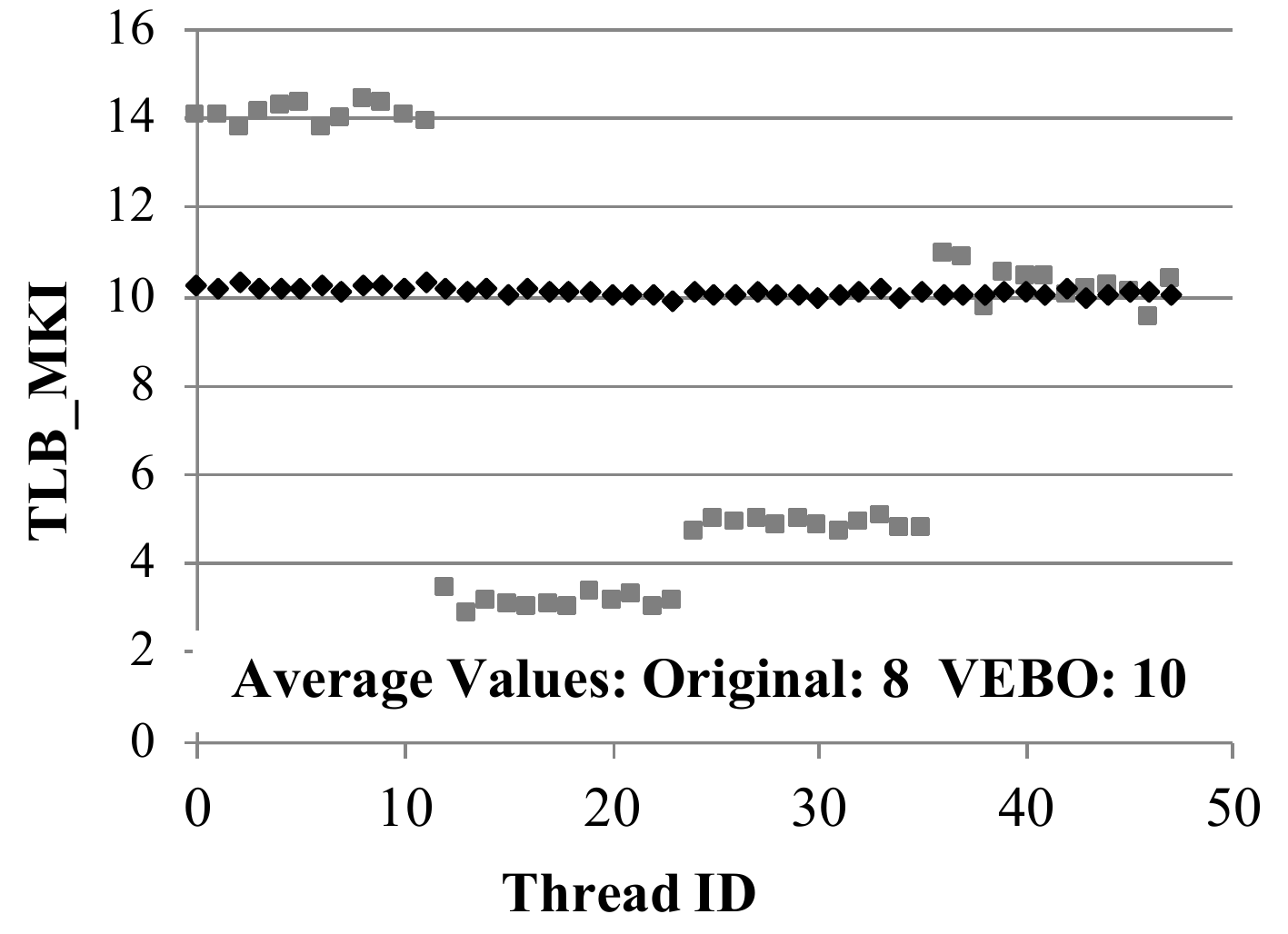}
		\label{fig:pr-tlb}
		\vspace*{-8mm}
	}
	\subfloat[PR -- branches]{
		\hspace*{-4mm}
		\includegraphics[width=.19\textwidth,clip]{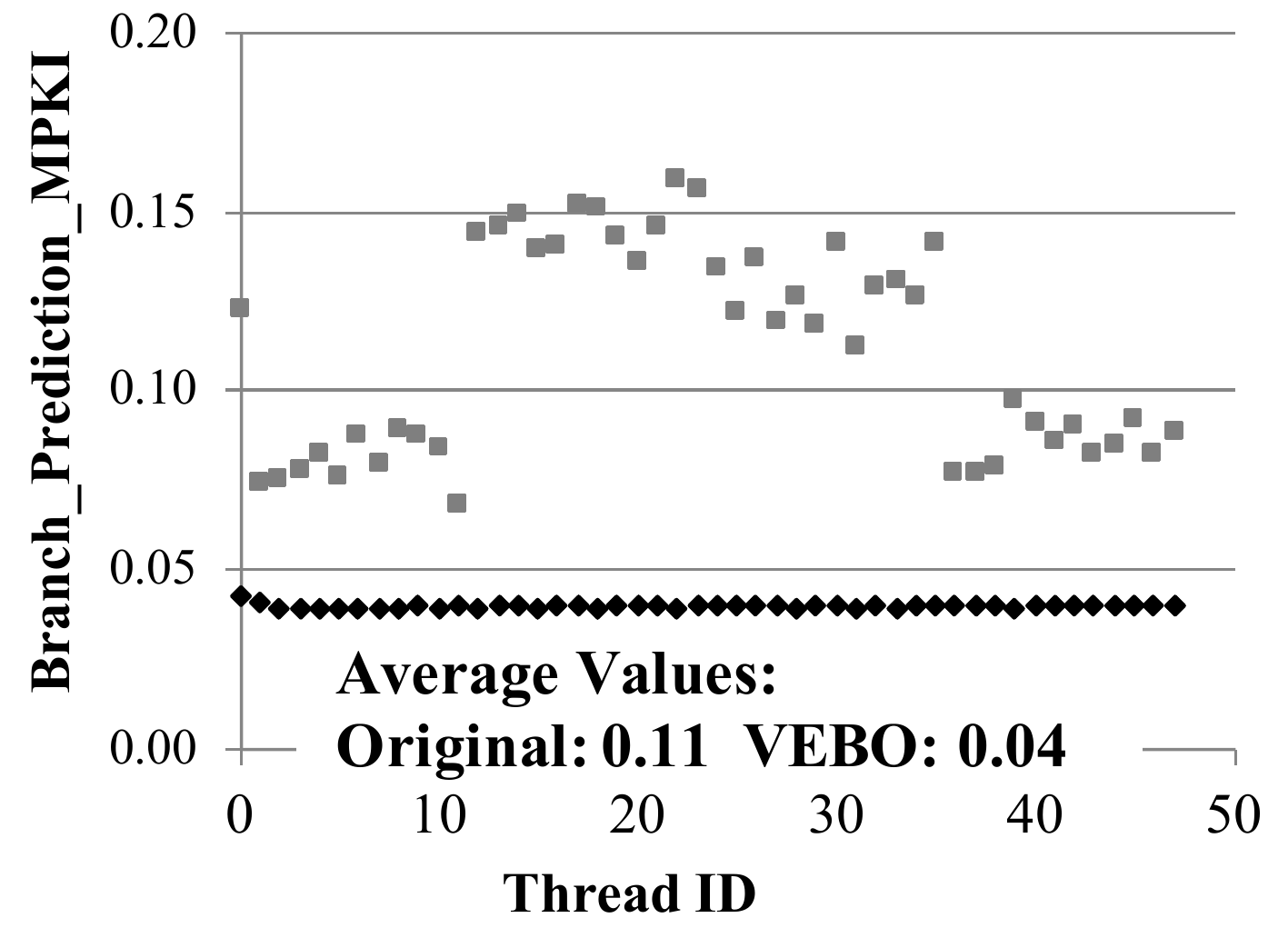}
		\label{fig:pr-branch}
		\vspace*{-8mm}
	}
	
	\caption{Execution time and micro-architectural statistics per partition or per thread for PR with Twitter. Measured on GraphGrind using 384 partitions. Thread $t$ executes partitions $8\:t$ to $8\:t+7$. Architectural statistics expressed in misses per thousand instructions (MPKI). }
	\label{fig:numatimer}
\end{figure*}
\subsection{Random Permutation of Graphs}
We evaluated the execution time of GraphGrind with four orders of vertex IDs: (a)~original vertex IDs of a graph, (b)~VEBO applied in original vertex IDs (c)~a random permutation of the vertex IDs, and (d)~VEBO applied to this random permutation.
Evaluation on the Twitter and USAroad graphs (other graphs show similar trends) shows that the random permutation has higher execution time compared to the other orders (Figure~\ref{fig:random}).
This occurs as a random permutation creates load imbalance. It also removes
locality that may exist result from collecting the graph~\cite{mcsherry:05:pagerank}.
VEBO has lower execution time compared to the random permutations
which demonstrates that VEBO is a sound algorithm and cannot be beaten
easily by any permutation of vertices.
Moreover, VEBO applied to the random permutation corrects load balance
and restores performance to nearly the same level as VEBO applied to the original graph. The difference in performance, if any, can be attributed to
differences in locality, which VEBO does not optimize.
	\begin{figure}[t]
		\vspace*{-5mm}
	\centering	
	\subfloat[Twitter]{
		\includegraphics[width=.48\columnwidth,clip]{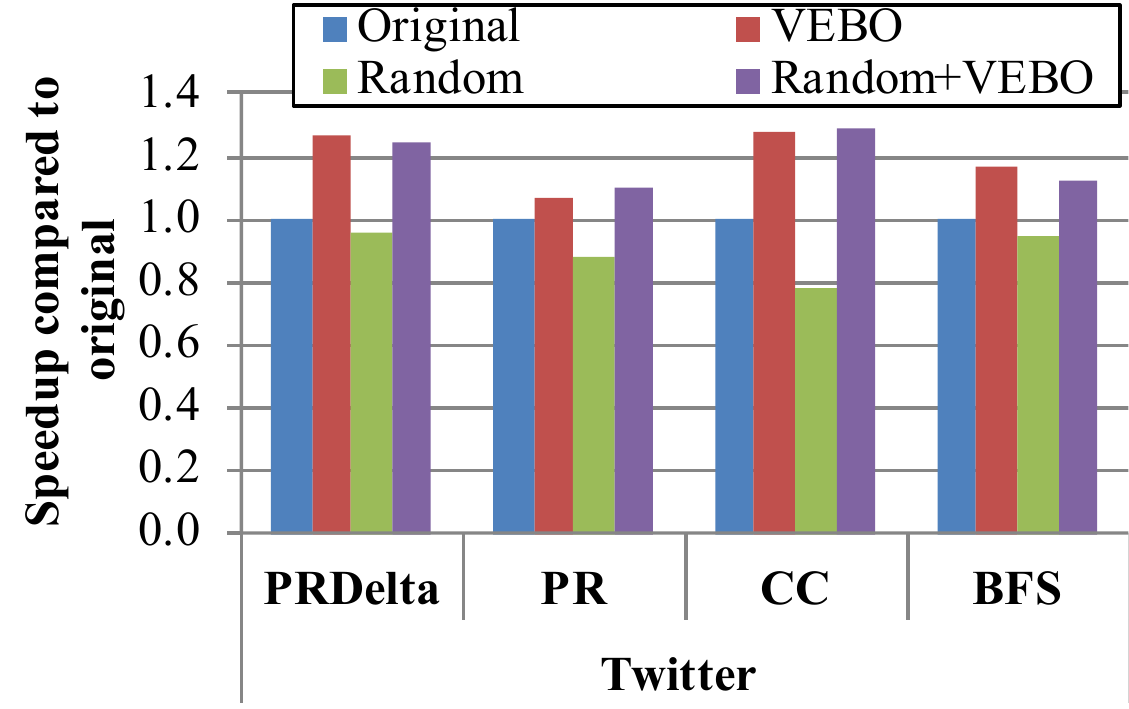}
		\label{fig:twitter_random}
	}
	\subfloat[USAroad]{
		\includegraphics[width=.49\columnwidth,clip]{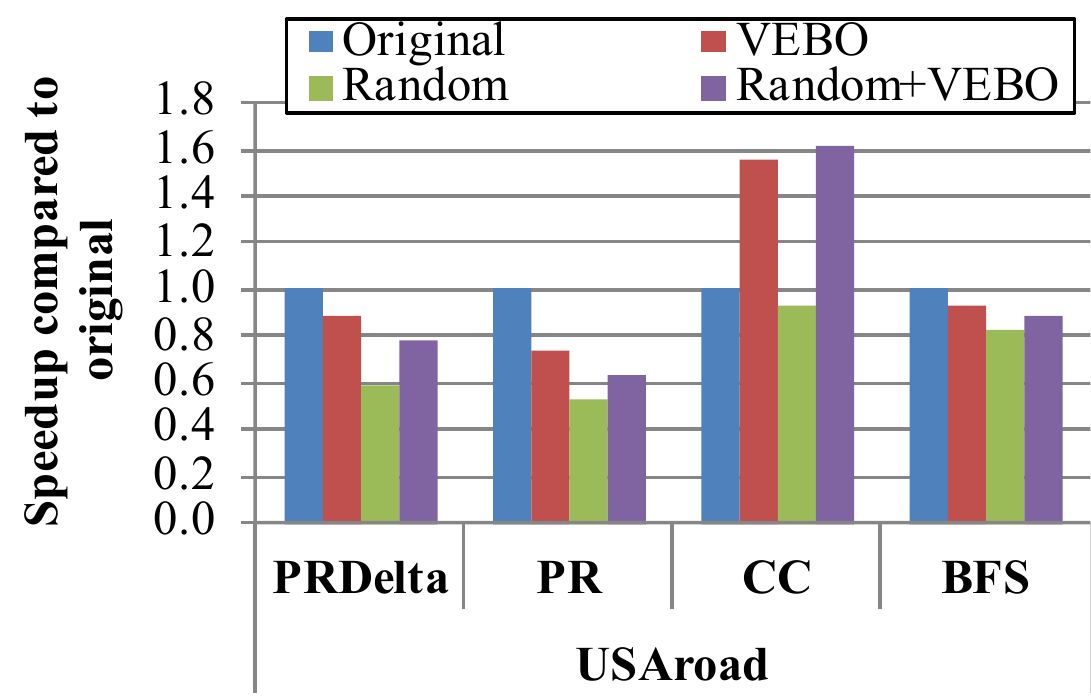}
		\label{fig:USAroad_random}
	}
	
	\caption{Performance of BFS on the Twitter and USAroad graphs using original vertex IDs, VEBO applied to original vertex IDs, a random permutation of the vertex IDs and VEBO applied to random vertex IDs. Execution times are normalised to the left-most case (original vertex IDs).}
	\label{fig:random}
\end{figure}	

\subsection{Sparse Frontier}
\begin{table}[t!]
  \centering
\caption{Distribution of active edge over partitions for the sparse iteration in BFS of Twitter graph using 384 partition. S.D. is standard deviation value. Active Edge/Part is the ideal number of active edge per partition.}
\scriptsize
    \begin{tabular}{|c|c|c|c|c|c|c|}
    \hline
    \multicolumn{2}{|c|}{\textbf{Iteration}} & \textbf{3} & \textbf{4} & \textbf{5} & \textbf{6} & \textbf{7} \\
    \hline
    \multicolumn{2}{|c|}{\textbf{Active Edge}} & 7986019 & 6872249 & 636055 & 54173 & 5926 \\
    \hline
    \multicolumn{2}{|c|}{\textbf{Active Edge/Part}} & 20797 & 17896 & 1656  & 141.1 & 15.43 \\
    \hline
    \multirow{2}[4]{*}{\textbf{Min}} & \textbf{Orig.} & 0     & 0     & 0     & 0     & 0 \\
\cline{2-7}          & \textbf{VEBO} & 0     & 2533  & 194   & 5     & 0 \\
    \hline
    \multirow{2}[4]{*}{\textbf{Median}} & \textbf{Orig.} & 541.5 & 13025 & 742.0 & 46.00 & 0.000 \\
\cline{2-7}          & \textbf{VEBO} & 10114 & 16448 & 1486  & 127.0 & 11.00 \\
    \hline
    \multirow{2}[4]{*}{\textbf{S.D.}} & \textbf{Orig.} & 79964 & 20539 & 3175  & 249.9 & 34.78 \\
\cline{2-7}          & \textbf{VEBO} & 60525 & 16571 & 2462  & 170.5 & 25.93 \\
    \hline
    \multirow{2}[4]{*}{\textbf{Max}} & \textbf{Orig.} & 1114127 & 148077 & 50780 & 3688  & 445 \\
\cline{2-7}          & \textbf{VEBO} & 1112832 & 304976 & 44647 & 3216  & 370 \\
    \hline
    \end{tabular}%
  \label{tab:sparse_edge}%
\end{table}%


Table~\ref{tab:sparse_edge} shows distribution of active edge per iteration of BFS when using 384 partitions (distribution of active destination shows similar trends). Sparse frontiers vary strongly in size. The active edges per partitions is the optimization target as it corresponds to perfectly balanced partitions. VEBO distributes both high-degree and low-degree vertices uniformly over partitions. Compared to VEBO, original has many partitions with zero active degree. Iterations 3-7 dominate execution time during which the load balance is significantly improved by VEBO. VEBO reduces standard deviation up to 1.5x and it reduces the gap between minimum and maximum number of active edge per partition.

\subsection{Analysis of Load Balance}
\label{sec:mpki}
Figure~\ref{fig:motiv} in Section~\ref{sec:motiv} shows
that VEBO balances edge and vertex counts.
We show here that this load balance translates to run-time statistics.
We focus on the PR algorithm for Twitter graph
(Figure~\ref{fig:numatimer}); the Friendster graph is similar, but
not shown for brevity.
We perform this analysis using GraphGrind.

Figure~\ref{fig:pr-time} shows the execution time for
each of the 384 partitions. There is a large variation on the
execution time for the original graph, e.g., from 0.290$\:$s per iteration
to 1.985$\:$s. 
For the VEBO reordered graph (darkest symbols),
the worst-case difference between the fastest and slowest partition
is 0.17$\:$s or more than 10 times less than the original graph.
For most graphs and algorithms, the average
time to process a partition is reduced by VEBO.
This contributes to the speedup caused by VEBO, besides
achieving load balance.

VEBO may reduce each partition's performance but this is compensated
by improved load balance.
For instance, when processing PR for Twitter, the average
execution time per partition is 1.211$\:s$ for VEBO
with a 1.6x spread
and 1.221$\:s$ for the original graph with a 6.9x spread.

VEBO balances execution at the micro-architectural level, namely
miss rates for caches, TLBs and branch predictors (Figure~\ref{fig:numatimer}).
Moreover,
we observed that VEBO improves memory locality for the majority of the graphs,
as the cache and TLB statistics are reduced.
PR for Twitter, however, is a rare counter-example.

Finally, VEBO reduces the branch misprediction rate
(Figures~\ref{fig:pr-branch}).
We attribute this in part to ordering vertices by decreasing degree.
When traversing
the compressed sparse rows (CSR)
and compressed sparse columns (CSC) data structures,
a loop iterates over the edges incident to a vertex. The loop iteration
count is determined by the degree.
In the VEBO graph, subsequent vertices have the same degree
which makes this branch highly predictable.
In the original graph, subsequent vertices
have highly varying degrees, which makes it hard to predict the loop
termination accurately.


\subsection{Edgemap vs Vertexmap}
\begin{table}[t!]
	\centering
	\footnotesize
	\caption{Architectural events for \emph{vertexmap} and \emph{edgemap}:
          cache misses serviced from the local NUMA node, from the remote NUMA node (Rmt)
          and TLB misses. Numbers expressed as MPKI. }
	\begin{tabular}{|l|p{0.35cm}|p{0.7cm}|c|c|c||c|c|c|}
		\hline
		& \multirow{2}[2]{*}{App.} & \multirow{2}[2]{*}{Order} & \multicolumn{3}{c||}{Vertex Map} & \multicolumn{3}{c|}{Edge Map} \\
		\cline{4-9}          &       &       & Local & Rmt & TLB &  Local & Rmt & TLB   \\
		\hline
                \parbox[t]{2mm}{\multirow{4}{*}{\rotatebox[origin=c]{90}{Twitter}}}
		& \multirow{2}[4]{*}{PR} & Ori. & 4.5 & 4.1 & 0.02 	 & 11.1 & 9.3 & 8.3 \\
		\cline{3-9}          &       & VEBO  & 6.9 & 1.6 & 0.01 	 & 12.0 & 12.2 & 9.4\\
		\cline{2-9}          & \multirow{2}[4]{*}{BF} & Ori. & 2.5 & 2.0 & 0.03 	 & 9.1 &11.0 & 11.5 \\
		\cline{3-9}          &       & VEBO  & 3.6 & 0.5 & 0.01 	 & 8.9 & 10.6 & 11.2 \\
		\hline
                \parbox[t]{2mm}{\multirow{4}{*}{\rotatebox[origin=c]{90}{Friendster}}}
		& \multirow{2}[4]{*}{PR} & Ori. & 8.3 & 3.3 & 0.01 	 & 33.0 & 28.7 & 34.8 \\
		\cline{3-9}          &       & VEBO  & 9.0 & 2.2 & 0.008 	 & 21.4 & 19.3 & 10.1 \\
		\cline{2-9}          & \multirow{2}[4]{*}{BF} & Ori. & 6.0 & 1.5 & 0.02  & 27.0 & 20.5 & 23.6 \\
		\cline{3-9}          &       & VEBO  & 6.6 & 0.8& 0.01& 22.6 & 16.7 & 20.2 \\
		\hline
	\end{tabular}%
	\label{tab:cache_table}%
\end{table}%

Graph algorithms are expressed by
means of \emph{edgemap} and \emph{vertexmap} traversals.
VEBO simultaneously balances edges and vertices and so load balances
both edgemap and vertexmap.
The performance benefits are, however, different.
Table~\ref{tab:cache_table} shows the summary statistics across
the edgemap and vertexmap operations for Twitter and Friendster with PageRank (PR) and Bellman-Ford (BF).
These statistics are collected per thread and correspond to the
execution of 8 consecutive partitions.
Edgemap generally dominates the execution time, so these statistics
correspond closely to Figure~\ref{fig:numatimer}.
Local and remote cache misses
as well as TLB misses are significantly reduced,
except of PR for Twitter.
VEBO generally improves memory locality during
edgemap, even though this was not
part of the optimization criterion.


Vertexmap benefits from load balancing rather than locality.
GraphGrind spreads the iterations of the vertexmap loop equally
across all threads~\cite{sun:17:ics}.
Arrays accessed by vertexmap, however, are distributed over the NUMA nodes
according to the graph partitions.
This causes a high number of remote cache misses because
Algorithm~\ref{algo:part} induces imbalance in the number
of vertices per partition.
VEBO ensures that all partitions have an equal number
of vertices. As such, each thread mostly accesses NUMA-local data,
which explains the reduction in remote misses.
\subsection{Space Filling Curves}
\label{sec:sfc}
The dense frontiers of GraphGrind use the COO, which could be the same as if using CSR or CSC. Alternatively, ordering edges using the Hilbert space filling curve gives a significant performance boost~\cite{sun:17:icpp}.
However Hilbert order is a heuristic, which has been studied
mostly for dense matrix algebra~\cite{ding:99:sfc,mellor-crummey:01:sfc}.
We have found that 
(i)~there exist cases where Hilbert order degrades performance;
(ii)~the effectiveness of Hilbert order depends on the number of non-zeroes.
To the best of our knowledge,
neither of these properties are adequately covered in the literature.

We sort all vertices from high to low degree and
partition the resulting graph into 384 partitions using Algorithm~\ref{algo:part}.
We compare the performance of this high-to-low order against VEBO for PageRank (Figure~\ref{fig:Hilbert-1-384}).
The first partitions contain the vertices with the highest degrees, which are processed faster than a partition with a mix degrees of VEBO.
The last partitions contain exclusively
degree-one vertices and are processed up to three times slower
than VEBO. 
This demonstrates that Hilbert order is more effective when the
in-degree of vertices is high. High degrees imply more opportunity
for reuse of data, which admit memory access order optimization.



Next,
we compare Hilbert order to the
traversal order of CSR,
i.e., by increasing source vertex ID (Figure~\ref{fig:Hilbert-csr}).
Surprisingly,
the CSR order admits faster processing for partitions 0--350 (approximately).
Thus, for high-degree vertices the CSR order is more efficient than Hilbert order.

As VEBO creates nearly the same degree distribution in each partition,
it is expected that CSR order is more efficient than Hilbert order.
We have modified GraphGrind accordingly to change COO using CSR order.
This change consistently
speeds up the algorithms with dense frontiers.
\begin{figure}[t]
	\centering	
    \vspace*{-5mm}
	\subfloat[High-to-low and VEBO]{
		\begin{minipage}{.48\columnwidth}
			\includegraphics[width=\columnwidth,clip]{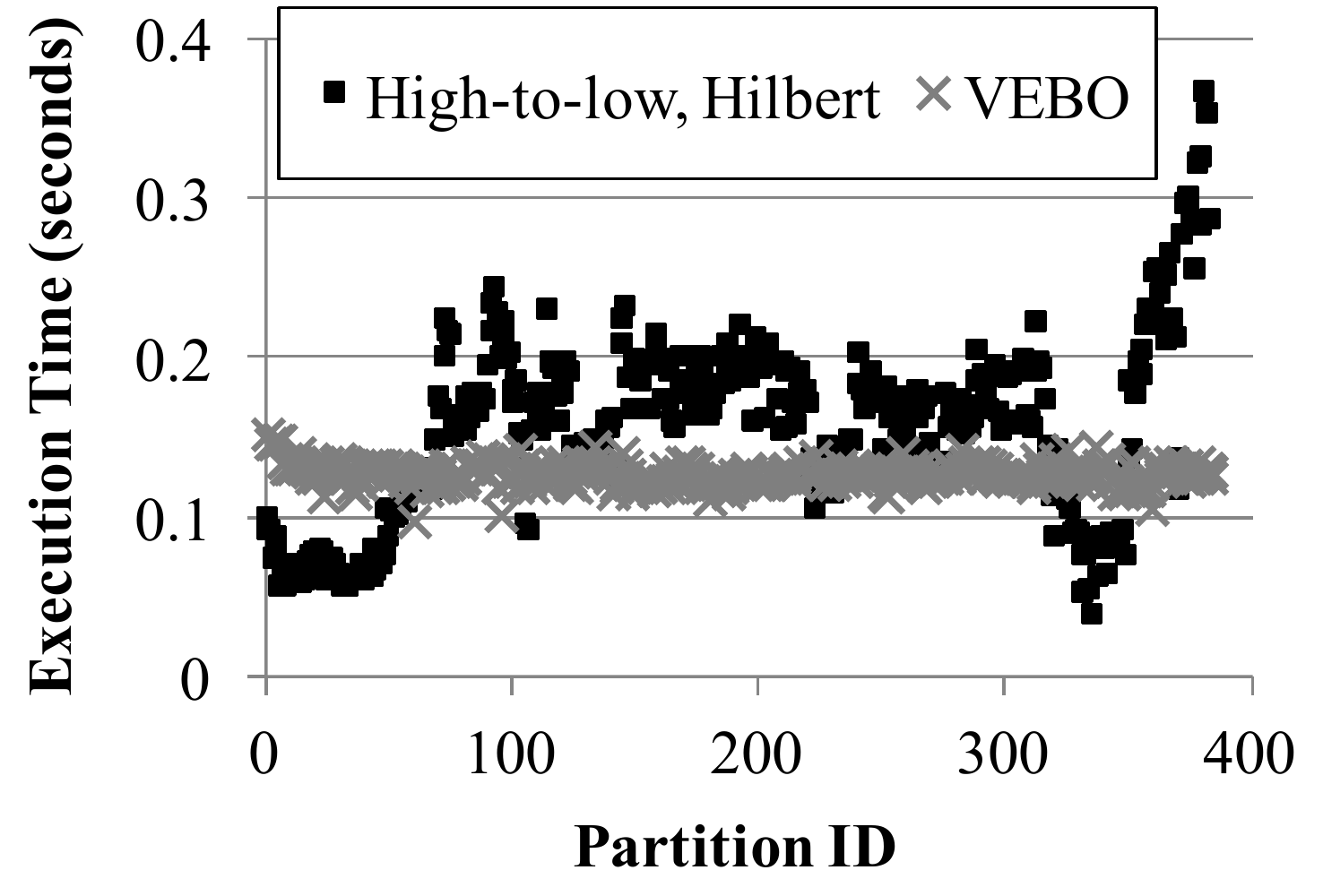}
		\end{minipage}
		\label{fig:Hilbert-1-384}
	}
	\subfloat[Hilbert and CSR]{
		\begin{minipage}{.48\columnwidth}
			\includegraphics[width=\columnwidth,clip]{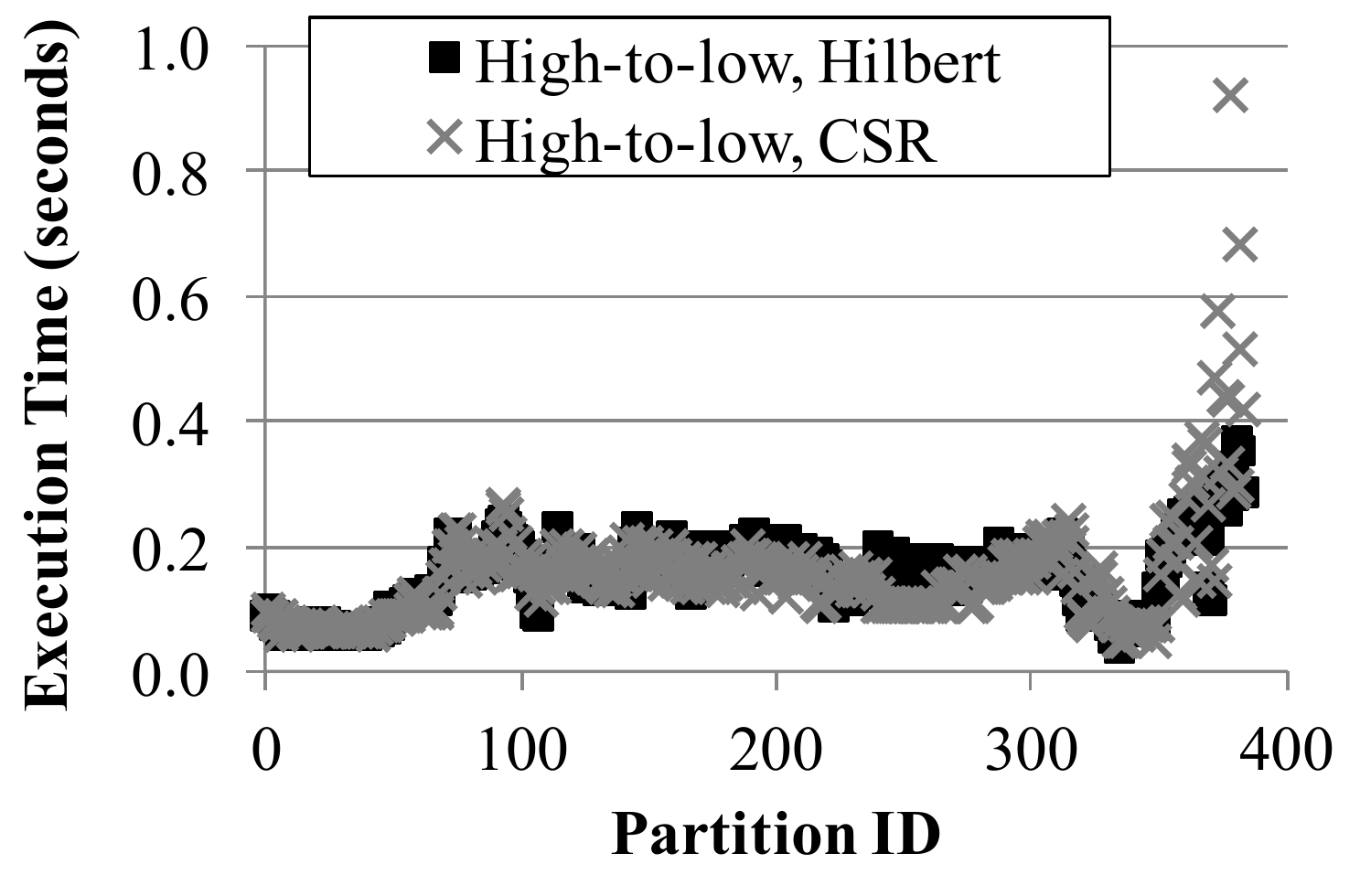}
		\end{minipage}
		\label{fig:Hilbert-csr}
	}
	\caption{Processing speed as a function of the in-degree.
		First iteration of PR on Twitter shown.}
\end{figure}
\begin{table*}[t!]
	\centering
	\footnotesize
	\caption{Overhead in seconds of vertex reordering, edge reordering and partitioning for Twitter and Friendster compared to PR and BFS execution time. Fastest results in bold.}
	\begin{tabular}{|c||c|c|c||c|c||c|c||c|c|}
		\hline
		\multirow{2}[4]{*}{Graph}& \multicolumn{3}{c||}{Vertex reordering} & \multicolumn{2}{c||}{Edge reordering + partitioning} & \multicolumn{2}{c||}{BFS} & \multicolumn{2}{c|}{PR (50 iterations)}  \\
		\cline{2-10}      &    RCM   & Gorder &VEBO & Hilbert order & CSR order & Original & VEBO & Original & VEBO \\
		\hline
		\hline
		Twitter   & 519.558 & 7803.871 & \textbf{5.119} & 10.684 & \textbf{4.412 }& 0.245 & \textbf{0.210} & 59.122 & \textbf{48.465}\\
		\hline
		Friendster & 755.384 & 8930.228 & \textbf{11.164} & 13.482 & \textbf{5.441} & 0.931 & \textbf{0.513} & 147.220 & \textbf{76.532} \\
		\hline
	\end{tabular}%
	\label{tab:overhead}%
\end{table*}%

\subsection{Overhead of Vertex Ordering}
Table~\ref{tab:overhead} shows that VEBO reduces the cost of vertex ordering up to 101x over RCM and 1524x over Gorder.
Moreover, VEBO performs best with CSR order for COO (Table~\ref{tab:order}).
This way of edge reordering and partitioning is faster than
using Hilbert edge order:
from 10.7s down to 4.4s for Twitter and 13.5s to 5.4s for Friendster.
While graph preparation incurs some overhead, the overhead is far less than
the gain. E.g., PR for Friendster takes 13.5s+147.2s on GraphGrind,
while it takes 11.2s+5.4s+76.5s on GraphGrind with VEBO (PR typically requires
over 50 iterations to converge).
BFS traverses graphs from a single vertex, and has a significantly shorter
execution time. However, in any reasonable setup multiple graph analytics
would be performed each time a graph is loaded in memory
such that the overhead of vertex and edge reordering can be amortized.


\section{Related work}
\label{sec:rela}
Graph partitioning has been thoroughly investigated.
Graph partitioning problem is formulated
as calculating a subset of the edges (or vertices) such that
the number of edges crossing partitions is minimized~\cite{stanton2012streaming,zhang2015numa,andreev:06:bgp,feige:02:polylog,tsourakakis:14:fennel}.
Additionally, authors specify a constraint to balance the edges~\cite{karypis:98:multilevel}
or vertices~\cite{bourse:14:bep,gonzalez2012powergraph}.


The exact solution to the graph partitioning problem
is NP complete (e.g.,~\cite{feige:02:polylog}).
Many authors have considered approximate algorithms
to achieve a close-to-optimal solution in polynomial time~\cite{andreev:06:bgp,feige:02:polylog}.
These may produce partitions of similar quality as general-purpose
graph partitioners such as METIS~\cite{karypis:98:metis}
in less time~\cite{feige:02:polylog}.

One may partition the edge set or the vertex set.
Partitioning the vertex set leads to a problem
of minimizing the edge cut, potentially under a constraint of
edge balance~\cite{kyrola2012graphchi,roy2013x,zhang2015numa}.
Partitioning the edge set leads to better heuristics and higher-performing
implementations~\cite{gonzalez2012powergraph}.
Edges now belong to a partition, while vertices may be replicated.
In this problem, rather than minimizing the edge cut,
the optimization criterion is to minimize the amount of vertex replication,
also known as vertex cut.

Different graph partitioning approaches are used in distributed memory
systems vs shared memory systems. Communication cost dominates
distributed memory systems, hence edge cut and vertex replication
are minimized~\cite{gonzalez2012powergraph,chen:15:powerlyra}.
In shared memory systems, the best performing systems ensure that
each partition contains vertices with consecutive vertex IDs.
This simplifies indexing and improves memory locality.
As such, partitioners such as METIS are not immediately applicable
and additional vertex relabeling must be applied.

Streaming partitioning algorithms partition the graph in a single pass
using a limited amount of storage~\cite{stanton2012streaming,tsourakakis:14:fennel}.
These algorithms compute approximations to the optimal partition
of similar quality to METIS~\cite{karypis:98:metis} in
a fraction of the time~\cite{tsourakakis:14:fennel}.
Gonzalez~\emph{et al} proposed \emph{vertex cut},
a parallel streaming partitioning algorithm that
minimizes vertex replication~\cite{gonzalez2012powergraph}.
Li~\emph{et al}~\cite{li:17:spac} and Bourse~\emph{et al}~\cite{bourse:14:bep} proposed efficient edge-balanced partitioning methods.
Bourse~\emph{et al}~\cite{bourse:14:bep} moreover investigate the interplay between edge balance and vertex balance, which is non-trivial if edge cuts are simultaneously minimized.

PowerLyra~\cite{chen:15:powerlyra} differentiates "high-degree" vertices from "low-degree" vertices and applies different partitioning methods. It aims to minimize the replication factor. VEBO is different. It explicitly avoids minimizing replication factor and edge cut as this is computationally demanding. It is likely that VEBO can further improve PowerLyra because it is easier to minimize the edge cut when the high-degree vertices are processed first.

Graphs may be stored in multiple, equivalent ways.
Vertex reordering aims to exploit the degree of freedom in vertex IDs.
\textbf{Gorder}~\cite{wei:16:gorder} proposes a general vertex ordering approach to improve CPU cache utilization. They find an optimal permutation among all nodes of a graph that retains temporal locality for nodes that are frequently accessed together.
The \textbf{RCM} algorithm reduces the bandwidth of sparse matrices by relabeling vertices~\cite{george:94:sparse}. 
\textbf{SlashBurn}~\cite{lim2014slashburn} exploits the hubs and their neighbours to define an alternative community different from the traditional community.
\textbf{LDG}~\cite{stanton2012streaming} is a heuristic streaming partitioner for large distributed graphs.

Edge reordering changes the order of edge traversal.
Switching between CSR and CSC~\cite{beamer:12:bfs}
is an edge reordering optimization.
Space-filling curves tend to increase temporal locality~\cite{murray:13:naiad}.
Extensive partitioning (a.k.a.\ segmentation~\cite{zhang:16:cache})
of CSR and CSC representations also improves
temporal locality~\cite{sun:17:icpp}.
The interaction between vertex and edge reordering is not covered well in the
literature.

\section{Conclusion}
\label{sec:concl}
The established heuristic to balance the processing time of graph partitions
is to create edge-balanced partitions.
We have demonstrated that edge-balance alone does
not create good load balance and that considering vertex-balance along
with edge-balance improves load balance significantly.
Moreover, our results show that minimizing edge cut or vertex replication
is not necessary on shared memory systems, and by extension on
shared memory subsystems in distributed graph processing.
We design VEBO, a vertex reordering algorithm for joint vertex and
edge balancing and demonstrate that it achieves excellent load balance.
for graphs with a power-law degree distribution.

We experimentally evaluated the performance of VEBO on three shared-memory
graph processing systems: Ligra, Polymer and GraphGrind.
Graph processing systems using static scheduling, such as Polymer and
GraphGrind, benefit more strongly from VEBO.

In future work, we will investigate whether distributed
graph processing systems, which typically use static scheduling,
also benefit from increased load balance even if this comes at the expense
of a small increase in vertex replication, and thus an increase in the
volume of data communication.
%
While VEBO improves load balance,
there still remain unknown factors
that affect the processing time of a graph partition.
Identifying those factors may lead to still higher efficiency.



	
        
	\bibliographystyle{IEEEtran}
	
	\bibliography{references}






\appendix

\section{Artifact Description Appendix: VEBO: A Vertex- and Edge-Balanced Ordering Heuristic to Load Balance Parallel Graph Processing}

\subsection{Abstract}

This description contains the information needed to launch some experiments of the paper "VEBO: A Vertex- and Edge-Balanced Ordering Heuristic to Load Balance Parallel Graph Processing". We explain how to compile and run the modified VEBO, GOrder and RCM in Ligra, Polymer and GraphGrind examples used in section IV. The results from section V are produced using NUMA-aware clang compiler, but this artifact description is not focused on that part of the paper.
\subsection{Description}

\subsubsection{Check-list (artifact meta information)}


{\small
\begin{itemize}
  \item {\bf Algorithm}: Graph ordering algorithm VEBO
  \item {\bf Program}: C++ code with Cilkplus extension
  \item {\bf Compilation}: icpc 14.0.0 and clang 3.4.1.
  \item {\bf Data set}: Publicly available graph files in adjacency format.
  \item {\bf Run-time environment}: Linux version 3.10.0-229.4.2.el7.x86\_64
  \item {\bf Hardware}: An x86-64 NUMA system.
  \item {\bf Output}: VEBO generates reordered graphs with adjacency format and timing measurements (wall clock time) on data loading time and reordering time.
  \item {\bf Experimental workflow}: Graph data sets are reordered with VEBO prior to loading in public graph processing frameworks (Ligra, Polymer and GraphGrind).
  \item {\bf Publicly available?}: Yes, after publication of paper
\end{itemize}
}

\subsubsection{How software can be obtained (if available)}
VEBO will be shared under open source license upon acceptance of the paper.

\subsubsection{Hardware dependencies}
We use a 4-socket 2.6GHz Intel Xeon E7-4860 machine with 256GB of DRAM
in our experiments.

\subsubsection{Software dependencies}
VEBO is a stand-alone tool.
Experiments make use of Ligra (\url{https://github.com/jshun/ligra}),
Polymer (\url{http://ipads.se.sjtu.edu.cn:1312/opensource/polymer.git}), and
GraphGrind (\url{https://github.com/Jaiwen/GraphGrind}).
We compare against item GOrder and RCM
(\url{https://github.com/datourat/Gorder}).

Ligra requires Cilkplus;
GraphGrind requires a custom version of Cilkplus with NUMA extension.
We have used the customized clang (\url{https://github.com/hvdieren/swan_clang}),
LLVM (\url{https://github.com/hvdieren/swan_llvm}),
and Cilkplus runtime (\url{https://github.com/hvdieren/swan_runtime}).

\subsubsection{Datasets}
\begin{itemize}
	\item Friendster, Orkut, LiveJournal are from Stanford Network Analysis Platform (SNAP). (\url{http://snap.stanford.edu/snap/})
	\item Powerlaw graph is generated by snap-standford graph generator.(\url{https://github.com/snap-stanford/snap/tree/master/examples/ graphgen})
	\item RMAT27 graph is generated by Problem Based Benchmark Suite.(\url{http://www.cs.cmu.edu/~pbbs/})
	\item Twitter is a social network graph from "What is Twitter, a social network or a news media?"~\cite{kwak2010twitter}.
	\item Yahoo\_mem is from Yahoo! Inc. (\url{http://webscope.sandbox.yahoo.com})
	\item USAroad is from 9th DIMACS Implementation Challenge - Shortest Paths. (\url{http://www.dis.uniroma1.it/challenge9/data/USA-road-d/USA-road-d.USA.gr.gz})
\end{itemize}
\subsection{Installation}
Download VEBO and compiler using Cilkplus compiler, e.g.,
\begin{itemize}
	\item clang++ 4.9.2 or higher with support for Cilkplus.
	\item Intel icpc compiler 
\end{itemize}
Compiling VEBO:
icpc -O3  -fcilkplus -g -c rMatGraph.C -o rMatGraph.o


\subsection{Experiment workflow}
After downloading the package from XXXXX, install it using the above instruction. For reordering, run the tool using the following command:
./VEBO -r 100 -p 384 original vebo
Where:
\begin{itemize}
	\item r: start vertex to track in the new reordering graph
	\item p: number of graph partitions.
	\item original: file containing the original graph (input)
	\item vebo: file where reordered graph is stored
\end{itemize}
Output: A reordered graph using VEBO that is isomorphic to the original graph.

\subsection{Evaluation and expected result}
The expected result is that the partitions of the reordered graph have a balanced number of vertices (unique destinations) and edges in each parition, i.e., in each 1/384-th set of vertices.
It is expected that the reordered graph will be processed faster than the original graph with Polymer and GraphGrind when the graph is scale-free.

\subsection{Experiment customization}
There is no need for customization to produce the results in this paper.

\subsection{Notes}
None.


\end{document}